\documentclass[sigconf]{aamas} 

\usepackage{balance} 

\usepackage[prefix=s]{xcolor-solarized}
\usepackage{caption}
\usepackage{subcaption}

\setcopyright{none}
\renewcommand\footnotetextcopyrightpermission[1]{} 
\acmSubmissionID{104}

\title[Keeping the Harmony Between Neighbors: Local Fairness in Graph Fair Division]{Keeping the Harmony Between Neighbors: \\
Local Fairness in Graph Fair Division}

\author{Halvard Hummel}
\affiliation{
  \institution{Norwegian University of Science and Technology}
  \city{Trondheim}
  \country{Norway}}
\email{halvard.hummel@ntnu.no}

\author{Ayumi Igarashi}
\affiliation{
  \institution{University of Tokyo}
  \city{Tokyo}
  \country{Japan}}
\email{igarashi@mist.i.u-tokyo.ac.jp}

\begin{abstract}
We study the problem of allocating indivisible resources under the connectivity constraints of a graph $G$.
This model, initially introduced by Bouveret et al. (published in IJCAI, 2017), effectively encompasses a diverse array of scenarios characterized by spatial or temporal limitations, including the division of land plots and the allocation of time plots. In this paper, we introduce a novel fairness concept that integrates local comparisons within the social network formed by a connected allocation of the item graph. Our particular focus is to achieve pairwise-maximin fair share (PMMS) among the "neighbors" within this network. For any underlying graph structure, we show that a connected allocation that maximizes Nash welfare guarantees a $(1/2)$-PMMS fairness. Moreover, 
for two agents, we establish that a $(3/4)$-PMMS allocation can be efficiently computed. Additionally, we demonstrate that for three agents and the items aligned on a path, a PMMS allocation is always attainable and can be computed in polynomial time. Lastly, when agents have identical additive utilities, we present a pseudo-polynomial-time algorithm for a $(3/4)$-PMMS allocation, irrespective of the underlying graph $G$. Furthermore, we provide a polynomial-time algorithm for obtaining a PMMS allocation when $G$ is a tree.
\end{abstract}


\usepackage{booktabs}
\usepackage{tikz}
\usetikzlibrary{arrows,shapes, calc, fit, positioning}
\usepackage{mathtools}
\usepackage{enumitem}
\usepackage{todonotes}
\usepackage[capitalize]{cleveref}
\usepackage{algorithm}
\usepackage{algorithmic}
\usepackage{multirow}

\usepackage{thmtools}
\usepackage{thm-restate}
\usepackage{hyperref}

\usepackage{tikz}
\usetikzlibrary{arrows,decorations.pathreplacing, decorations.pathmorphing, positioning, calc}
\usepackage{caption}
\usepackage{subcaption}

\newcommand{\bbR}{\mathbb{R}}
\newcommand{\bbN}{\mathbb{N}}

\newcommand{\argmax}{\mbox{argmax}}
\newcommand{\ceil}[1]{\lceil #1 \rceil }
\newcommand{\floor}[1]{\lfloor #1 \rfloor }
\newcommand{\PMMS}{\mathrm{PMMS}}
\newcommand{\MMS}{\mathrm{MMS}}

\newcommand{\BibTeX}{\rm B\kern-.05em{\sc i\kern-.025em b}\kern-.08em\TeX}

\begin{document}


\pagestyle{fancy}
\fancyhead{}


\maketitle 


\section{Introduction}
Consider the distribution of offices among research groups. This task is not simple due to the inherent differences in size, quality, and location of the offices. Moreover, individuals may possess diverse preferences for these factors. For instance, certain individuals may prioritize office size, as they require sufficient space to conduct experiments, while others may emphasize proximity to student canteens. In light of these considerations, how can one effectively address individuals' needs and ensure a ``fair" distribution of offices?

This question has been explored in the extensive literature on fair division~\citep{Moulin03}. The existing literature has traditionally focused on two primary fairness notions: envy-freeness and proportionality. Envy-freeness requires each agent to compare their own bundle with that of everyone else and receive a preferred one.  Proportionality ensures that each agent receives a fair share, comprising at least $1/n$ of the total utility of the resource. However, achieving either of these notions becomes challenging when the resources are indivisible. For instance, when distributing a single office among two agents, we can satisfy neither of the fairness requirements. To address this limitation, recent research has explored alternative fairness measures that relax these requirements. Examples include envy-freeness up to one item (EF1) and maximin fair share (MMS)~\citep{Budish11}.

In this paper, we investigate a model that focuses on the \emph{connected} allocation of indivisible resources arranged in a graph structure. This model, as proposed by~Bouveret et al.~\cite{BouveretCeEl17}, captures various resource allocation scenarios that involve spatial or temporal constraints. Using the example of office allocation, it is not desirable for members of a research group to be assigned offices that are widely dispersed. The vertices of the graph can represent not only offices but also land plots, time slots, and other similar resources.

While envy-freeness and proportionality are natural fairness concepts within the graph fair division model, it may not be realistic to assume that individuals make global comparisons across the entire resource. According to the theory of social comparison, people often engage in local comparisons, where they assess their situation in relation to their peers, neighbors, or family members~\citep{Festinger1954,ScotPearlin,HandbookSulsWheeler}. 

In the context of fair division, several papers have explored this phenomenon by assuming that agents are part of a social network~\citep{Abebe2017,BeiGraph2017,BeynierGraph2019,BredereckGraph2022,HosseiniGraphical2023}. Much of this research focuses on local fairness requirements, where agents compare their allocation only with that of their neighbors within the underlying social network.

To integrate these two elements, fair division \emph{of} and \emph{over} graphs, we present a novel model that incorporates local comparison into the graph fair division setting. More precisely, we examine the social network formed by a connected allocation of the item graph and investigate local fairness among the ``neighbors" in the network. In this context, two agents are considered neighbors if they are allocated to vertices that are adjacent to each other. In practical scenarios, this could correspond to neighboring countries, employees assigned to consecutive shifts, or research groups with adjacent offices. 
Building upon recent research on approximate fairness~\citep{Budish11,caragiannis2019unreasonable}, we focus on the local variant of pairwise maximin fair share (PMMS). 
We aim at addressing the following question: 
\begin{quote}
\emph{What local fairness guarantees can be achieved under various types of graphs? Can we simultaneously attain local and global fairness?}
\end{quote}

\begin{table*}[htb]
\centering
\caption{Summary of our results for additive utilities. The non-existence~$\dagger$ holds due to a counter example for MMS existence on a graph~\cite{KurokawaPrWa18,BouveretCeEl17}. The existence~$\oplus$ follows from the cut-and-choose argument among two agents with additive utilities. The NP-hardness~$\otimes$ is due to the NP-hardness of \sc{Partition}.}
\small
\begin{tabular}{cccccc}
 \toprule
\emph{Properties} & & \emph{$n$ agents (any graph)} & \emph{$2$ agents (any graph)} & \emph{$3$ agents (path)} & \emph{Identical agents}\\
 \midrule
\multirow{2}{*}{MMS \& PMMS} & Existence & No$\dagger$ & Yes$\oplus$ & Open & Yes (Th.~\ref{lem:identical-utility-functions})\\\addlinespace[0.6pt]\cline{2-6}\addlinespace[0.6pt]
& Computation & NP-hard$\otimes$ & NP-hard$\otimes$ & Open & Poly for trees (Cor.~\ref{cor:PMMS+SMMS+trees})\\
\midrule
\multirow{2}{*}{$\alpha$-PMMS} & Existence & Yes for $\alpha=1/2$ (Th.~\ref{thm:MNW-half-PMMS}) & Yes for $\alpha=1$ $\oplus$ & Yes for $\alpha=1$ (Th.~\ref{thr:three-agents-path}) & Yes for $\alpha=1$ (Th.~\ref{lem:identical-utility-functions}) \\\addlinespace[0.6pt]\cline{2-6}\addlinespace[0.6pt]
& Computation & Open for constant $\alpha \in (0,1)$ & Poly for $\alpha=3/4$ (Th.~\ref{thm:halfPMMS-polynomial:twoagents}) & Poly  (Th.~\ref{thr:three-agents-path}) &  Pseudo-poly for $\alpha=3/4$ (Th.~\ref{thm:PMMS-pseudopolynomial})\\
\bottomrule 
\end{tabular}
\label{table}
\end{table*}

\subsection{Our contributions}
In Section~\ref{sec:model}, we present the formal model and introduce our local fairness concept of pairwise maximin fair share (PMMS), originally introduced by~Caragiannis et al.~\cite{caragiannis2019unreasonable}. Intuitively, it requires that each agent receives a fair share, determined by redistributing her bundle and the bundle of her neighbor via the the cut-and-choose algorithm. More precisely, the PMMS of agent $i$ with respect to her neighbor $j$ is the maximum utility $i$ can achieve if she were to partition the combined bundle of the pair into two and choose the worst bundle. While PMMS implies $(4/7)$-MMS in the standard setting of fair division~\cite{ijcai2018Amanatidis}, we show that there is no implication relation between PMMS and MMS up to any multiplicative factor in the graph-restricted setting. In \cref{sec:relation}, we further discuss the relationship between PMMS and other solution concepts. 

In Section~\ref{sec:nagents}, we establish the strong compatibility between the local fairness notion of PMMS and the global criterion of Pareto-optimality. We show that for $n$ agents with additive utilities and general graphs, any MNW allocation that maximizes the Nash welfare (i.e., the product of utilities) satisfies PMMS up to a factor of $0.5$. This stands in sharp contrast to an impossibility result regarding EF1 that $\alpha$-EF1 with any $\alpha>0$ does not exist even for two agents and a star graph $G$; see Proposition~\ref{prop:alphaEF1}.
        
In Section~\ref{sec:twoagents}, we focus on the case of two agents with additive utilities. The two-agent case serves as a foundational model of the fair division problem and has wide-ranging practical applications, including divorce settlements and inheritance division~\cite{BramsFi00,brams2014two,kilgour2018two}. While computing a PMMS allocation is in general hard even for two agents, we show that $(3/4)$-PMMS allocation can be computed in polynomial time for any connected graph $G$. Our proof crucially rests on an important observation that for biconnected graphs, a $(3/4)$-PMMS allocation can be computed efficiently by constructing a bipolar number over the graph. By exploiting the acyclic structure of the maximal biconnected subgraphs (called blocks) of $G$, we show that our problem can in principle be reduced to that for some valuable block of $G$.\footnote{A similar technique using blocks has been used to characterize a family of graphs for which a connected allocation among two agents satisfying EF1 (and its relaxations) is guaranteed to exist~\cite{BiloCaFl22,BeiPrice}.}

In Section~\ref{sec:paths}, we consider the case of three agents with additive utilities. For three agents with non-identical additive utilities on a path, we show that a PMMS connected allocation exists and can be computed in polynomial time. Note that in the standard setting of fair division, 
the existence of PMMS allocations remains an open question, particularly for three agents with additive utilities. 

In Section~\ref{sec:identical}, we focus on the case of identical additive utilities. We show that when agents have identical utility, there exists a pseudo-polynomial time algorithm that computes a $(3/4)$-PMMS allocation of any connected graph. Further, if we add the condition that the underlying graph is a tree, we can compute an allocation that simultaneously satisfies both PMMS and MMS in polynomial time. In fact, our final allocation satisfies a stronger property of MMS where the number of agents receiving their exact maximin fair share is minimized. See Table~\ref{table} for an overview of our results. 

\smallskip
\noindent
{\bf Related work}
Our work aligns with the growing literature on fair allocation of and over graphs. 

Abebe et al.~\cite{Abebe2017} and Bei et al.~\cite{BeiGraph2017} initiated the study of fair allocation of divisible resources with \emph{agents} arranged on a social network. They introduced the concepts of local envy-freeness and local proportionality, which restrict comparisons to pairs of neighboring agents. Bredereck et al.~\cite{BredereckGraph2022} investigated local fairness in the context of indivisible resource allocation, providing computational complexity results for local envy-freeness. Beynier et al.~\cite{BeynierGraph2019} and Hosseini et al.~\cite{HosseiniGraphical2023} examined house allocation problems over social networks where each agent can be allocated at most one item. However, these works assume that the social network is predetermined, whereas in our work, the network structure is not predefined. 

In parallel, several papers have studied fair allocation of resources aligned on a graph~\citep{BouveretCeEl17,Greco:20,Igarashi2023,Igarashi_Peters_2019,Truszczynski:20,Lonc2023}. It has been shown that for well-structured classes of graphs, a connected allocation satisfying global fairness can always be achieved.  
For example, Bouveret et al.~\cite{BouveretCeEl17} demonstrated that for trees, a connected allocation satisfying maximin fair share (MMS) always exists. Another approximate fairness notion, envy-freeness up to one item (EF1), can be attained under connectivity constraints of a path~\citep{BiloCaFl22,Igarashi2023}. 
However, the existence guarantee does not hold for general graphs: an MMS connected allocation may not exist on cycles~\citep{BouveretCeEl17}, and an EF1 connected allocation is not guaranteed even for a star graph with two agents having identical binary utilities~\citep{BiloCaFl22}. 

Caragiannis et al.~\cite{caragiannis2019unreasonable} introduced the concept of PMMS and identified its connection with other solution concepts in the standard setting of fair division without connectivity constraints.  They showed that any PMMS allocation satisfies an approximate notion of envy-freeness, called EFX. Furthermore, they proved that any maximum Nash welfare (MNW) allocation always satisfies $\alpha$-PMMS where $\alpha \approx 0.618$ is the golden ratio conjugate. Amanatidis et al.~\cite{multiplebirds} improved this factor to $\frac{2}{3}$ and provided a polynomial-time algorithm for computing such an allocation. The current best multiplicative factor $\alpha$ for which $\alpha$-PMMS is known to exist in the standard setting of fair division is $0.781$ by Kurokawa~\cite{Kurokawa2017}. Amanatidis et al.~\cite{ijcai2018Amanatidis} discussed the relationship between approximate versions of PMMS and several fairness concepts, such as EF1 and EFX. See also \cite{fairsurvey2022} for a recent survey on fair division. 

\section{Model and fairness concepts}\label{sec:model}
We define a fair division problem of indivisible items where the items are placed on a graph. For a natural number $s \in \bbN$, we write $[s]=\{1,2,\ldots,s\}$. 
Let $N=[n]$ be a set of {\em agents} and $G=(V,E)$ an {\em item graph}. Throughout this paper, we assume that $G$ is a simple graph, i.e., $G$ does not contain loops or parallel edges. Each agent $i$ has a \emph{utility function} $u_i:2^V \rightarrow \bbR_{+}$. For simplicity, the utility of a single vertex $v \in V$, $u_i(\{v\})$, is also denoted by $u_i(v)$. The elements in $V$ are referred to as {\em items} (or \emph{vertices}). Each subset $X \subseteq V$ is referred to as a {\em bundle} of items. A bundle $X \subseteq V$ is {\em connected} if it is a connected subgraph of $G$. An {\em allocation} $A=(A_1,A_2,\ldots,A_n)$ is a partition of the items into disjoint bundles of items, i.e., $\bigcup_{i \in N}A_i=V$ and $A_i \cap A_j=\emptyset$ for every pair of distinct agents $i,j \in N$. We say that an allocation $A$ is \emph{connected} if for every $i \in N$, $A_i$ is connected in $G$. A utility function $u_i$ is \emph{monotone} if $u_i(X) \leq u_i(Y)$ for every $X \subseteq Y \subseteq V$. It is \emph{additive} if $u_i(X)=\sum_{v\in X} u_i(v)$ holds for every $X \subseteq V$ and $i\in N$. An additive utility $u_i$ is \emph{binary additive} if
$u_i(v)\in \{0,1\}$ for every $v \in V$ and $i \in N$. 

We introduce fairness notions including those based on local comparison among neighbors and those based on global comparison among agents. Given an allocation $A$, we say that $i,j \in N$ are {\em adjacent} under $A$ if there is an edge $\{v_1,v_2\} \in E$ such that $v_1\in A_i$ and $v_2 \in A_j$; we call $j$ a \emph{neighbor} of $i$. For every connected bundle $X \subseteq V$ and natural number $k$, we denote by $\Pi_k(X)$ the set of partitions of $X$ into $k$ connected subgraphs of a graph $G[X]$, where $G[X]$ is the subgraph of $G$ induced by $X$. 
For agent $i$, a connected bundle $X$, and natural number $k$, we define the $k$-\emph{maximin share} of agent $i$ with respect to $X$ as $\mu^k_i(X)=\max \{\, \min_{j \in [k]}u_i(A_j) \mid A \in \Pi_k(X)\,\}$. For $k=2$ and $k=n$, we write $\mu^2_i(X)$ and $\mu^n_i(X)$ as $\PMMS_i(X)$ and $\mathrm{MMS}_i(X)$, respectively. For agents with identical utility function $u$, we simply write $\PMMS(X)=\PMMS_i(X)$ and $\mathrm{MMS}(X)=\mathrm{MMS}_i(X)$ for each agent $i$. 

\begin{definition}[Pairwise MMS]
An allocation $A$ is called 
$\alpha$-PMMS if for every pair of agents $i,j \in N$ such that $A_i=\emptyset$ or $i$ and $j$ are adjacent under $A$, $u_i(A_i) \geq \alpha \cdot \PMMS_i(A_i \cup A_j)$. For $\alpha=1$, we refer to the corresponding allocations as PMMS allocations. 
\end{definition}

In order to avoid a trivial allocation (allocating all items to one agent) from satisfying PMMS, we allow agents with an empty bundle to compare their bundle with every other agent's bundle under the definition above. Note that under additive utilities we have $\frac{1}{2}u_i(X) \geq \PMMS_i(X)$, since in any partition of a set $X$ into two sets, one bundle always has a utility of at most a half of the total utility with respect $u_i$. 
We say that a partition $A=(A_1,A_2,\ldots,A_n)$ of $G$ is a \emph{PMMS partition} of agent $i$ if $A_j$ is connected for each $j \in [n]$, and it is PMMS for agent $i$ no matter which bundle $i$ receives. 
An allocation $A$ is called $\alpha$-MMS if for every agent $i \in N$, $u_i(A_i) \geq \alpha \cdot \mathrm{MMS}_i(V)$. For $\alpha=1$, we refer to the corresponding allocations as MMS allocations.

An allocation $A$ is $\alpha$-EF1 if for any pair of agents $i,j \in N$ with $A_j \neq \emptyset$, there exists an item $v \in A_j$ such that $u_i(A_i) \geq \alpha \cdot u_i(A_j \setminus \{v\})$. For $\alpha=1$, we refer to the corresponding allocations as EF1 allocations. 
Given an allocation $A$, another allocation $A'$ is a {\em Pareto-improvement} of $A$ if $u_i(A'_i) \ge u_i(A_i)$ for every $i \in N$ and $u_j(A'_j) > u_j(A_j)$ for some $j \in N$. 
We say that a connected allocation $A$ is {\em Pareto-optimal} if there is no connected allocation that is a Pareto-improvement of $A$. 

We say that a connected allocation $A$ is a \emph{maximum Nash welfare (MNW)} allocation if it maximizes the number of agents receiving positive utility and, subject to that, maximizes the product of the positive utilities, i.e.,  $\prod_{i\in N:\, u_i(A_i)>0} u_i(A_i)$, over all connected allocations. For an allocation $A$, let ${\bf u}(A)$ be the vector obtained from rearranging the elements of the vector $(u_1(A_1),u_2(A_2),\ldots,u_n(A_n))$ in increasing order. Given two allocations $A$ and $A'$, an allocation $A$ is a \emph{leximin improvement} of $A'$ if there exists $k \in [n]$ such that the first $k-1$ elements of ${\bf u}(A)$ and ${\bf u}(A')$ are the same, but the $k$-th element of ${\bf u}(A)$ is greater than that of ${\bf u}(A')$. A connected allocation $A$ is called \emph{leximin} if there is no connected allocation that is a leximin improvement of $A$. 

\subsection{Relationship between PMMS and other fairness notions}\label{sec:relation}
In this section, we consider the relationship between PMMS and other fairness notions. 
Note that there is no implication relation between PMMS and MMS even for a path and agents with identical additive utilities. 

\begin{restatable}{proposition}{propPMMSrelationMMS}\label{prop:relation:PMMS:MMS}
Even for a path and three agents with additive identical utilities, neither PMMS nor MMS implies the other up to any multiplicative factor $\alpha \in (0,1]$.
\end{restatable}
\begin{proof}
Consider any $\alpha \in (0,1]$. Choose $\beta$ such that $\alpha > \frac{1}{\beta}$. Note that $\beta >1$. 
To show that MMS may not satisfy $\alpha$-PMMS, consider an instance of three agents with utility function $u$ and four vertices on a path. Each agent has utility $1$, $\beta$, $\beta$, and $1$ for each of the vertices, starting from the leftmost. Consider an allocation $A$ that allocates the leftmost item to agent $1$, the rightmost item to agent $2$, and the remaining items to agent $3$. Here, the maximin fair share for each agent is $1$ and thus this allocation satisfies MMS. However, $A$ is not $\alpha$-PMMS. Indeed, we have $\frac{1}{\alpha}=\frac{1}{\alpha} u(A_1) < \PMMS(A_1 \cup A_2)=\beta$.

To show that PMMS may not satisfy $\alpha$-MMS, consider an instance of three agents with utility function $u$ and four vertices on a path. Each agent has utility $1$, $\beta$, $\beta$, and $\beta$ for each of the vertices, starting from the leftmost. Consider an allocation $A$ that allocates the leftmost item to agent $1$, the second leftmost item to agent $2$, and the remaining items to agent $3$. Here, the allocation satsfies PMMS. However, $A$ is not $\alpha$-MMS. Indeed, we have $\frac{1}{\alpha}=\frac{1}{\alpha} u(A_1) < \MMS(V)=\beta$.
\end{proof}

Next, we consider the relathionship between PMMS and envy-based notions. 
It turns out that neither of EF1 and PMMS is stronger than the other. 


\begin{proposition}
Even for a path and two agents with identical additive utilities, EF1 may not imply PMMS. 
\end{proposition}
\begin{proof}
Consider an instance of a path with three vertices and two agents with identical additive utilities. Each agent has utility $1$, $1$, and $10$ for each of the vertices from left to right. An allocation that assigns item $1$ to agent $1$ and the rest to agent $2$ is EF1; however it is not PMMS since transferring the middle item, with utility $1$, from agent $2$ to agent $1$ increases the minimum utility among the two agents. 
\end{proof}

\begin{proposition}
Even for a star and two agents with identical binary additive utilities, PMMS may not imply EF1. 
\end{proposition}
\begin{proof}
Consider an instance of a star with three leaf vertices and two agents, each of whom has utility $1$ for every vertex of the star. An allocation that allocates one leaf vertex to agent $1$ and the rest to agent $2$ is PMMS. However, it is not EF1 as agent $2$'s obtains a higher utility than agent $1$ even after removing a single vertex from agent $2$'s bundle. 
\end{proof}


The following proof is similar to that in \cite{ijcai2018Amanatidis}, who showed that EF1 implies $(1/2)$-PMMS in the standard setting of fair division.

\begin{proposition}
For a connected graph $G$ and $n$ agents with additive utilities, EF1 implies $(1/2)$-PMMS. 
\end{proposition}

\begin{proof}
Let $A$ be a connected allocation that is EF1. Consider any pair of agents $i, j$ such that $A_i=\emptyset$ or $i$ and $j$ are adjacent under $A$. Let $M_{ij}=A_i \cup A_j$. If $\PMMS_i(M_{ij}) = 0$ then $u_i(A_i) \ge 0 = \PMMS_i(M_{ij})$. Assume that $\PMMS_i(M_{ij}) > 0$. Since $\PMMS_i(M_{ij}) > 0$, each bundle must contain at least one item in a PMMS partition of $M_{ij}$ for $i$. As each item may only appear in one bundle, it holds that $\PMMS_i(M_{ij}) \le u_i((A_i \cup A_j)\setminus \{v\})$ for any $v \in A_i \cup A_j$. Since $A$ is locally EF1, $u_i(A_i) \ge u_i(A_j \setminus \{v\})$ for some $v \in A_j$. Thus, for this item $v$, we get
\[
\PMMS_i(M_{ij}) \le u_i((A_i \cup A_j)\setminus \{v\}) = u_i(A_i) + u_i(A_j \setminus \{v\}) \le 2u_i(A_i).
\]
Rewritten, we get the inequality $u_i(A_i) \ge \PMMS_i(M_{ij})/2$.
\end{proof}

When agents have binary additive utilities, it turns out that EF1 implies PMMS. 

\begin{proposition}\label{prop:binary:EF1andMMS}
For a connected graph $G$ and $n$ agents with binary additive utilities, 
EF1 implies PMMS.
\end{proposition}

\begin{proof}
    An allocation $A$ is PMMS under binary utilities if for any pair of agents $i$ and $j$, it holds that $u_i(A_i) \ge \floor{(u_i(A_i) + u_i(A_j))/2}$. In an EF1 allocation we have that $u_i(A_i) \ge u_i(A_j \setminus \{v\})$ for some $v \in A_j$. Since $u_{i}(v) \in \{0, 1\}$ it follows that $u_i(A_i) \ge (u_i(A_i) + u_i(A_j))/2 - 1/2$. Since $u_i(A_i), u_i(A_j) \in \mathbb{N}$, $\floor{(u_i(A_i) + u_i(A_j))/2} = (u_i(A_i) + u_i(A_j))/2 - 1/2$ if $u_i(A_i) + u_i(A_j)$ is odd. If $u_i(A_i) + u_i(A_j)$ is even, then it must hold that $u_i(A_i) \ge u_i(A_j)$ and $u_i(A_i) \ge \floor{(u_i(A_i) + u_i(A_j))/2}$. Otherwise we have that $u_i(A_j) - u_i(A_i) \ge 2$ and the allocation is not EF1.
\end{proof}

\begin{corollary}
For a path and $n$ agents with binary additive utilities, a PMMS and EF1 connected allocation exists.
\end{corollary}
\begin{proof}
The claim follows from \cref{prop:binary:EF1andMMS} and that a EF1 connected allocation exists on a path for monotone utilities~\cite{Igarashi2023}.
\end{proof}

Next, we consider EFX. An allocation $A$ is \emph{envy-free up to any item (EFX)} if for any pair of agents $i,j \in N$, $A_j = \emptyset$, or it holds that $u_i(A_i) \geq u(A_j \setminus \{v\})$ for every item $v \in A_j$. 

In the standard setting of fair division, it is known that PMMS implies EFX~\cite{caragiannis2019unreasonable}. We observe that the relation becomes reversed when the graph is a path. 

\begin{proposition}
For a path and $n$ agents with monotone utilities, EFX implies PMMS. 
\end{proposition}
\begin{proof}
Suppose that $A$ is an allocation that satisfies EFX. Assume towards a contradiction that it fails to satisfy PMMS. Thus, there exists a pair of agents $i,j \in N$ such that $v_1\in A_i$ and $v_2 \in A_j$ for some $\{v_1,v_2\} \in E$ and $u_i(A_i) < \PMMS_i(A_i \cup A_j)$. Let $(X,Y)$ be a PMMS partition of $A_i \cup A_j$ for agent $i$ where $X \cap A_i \neq \emptyset$. The above inequality means that $A_i \cup \{v_2\} \subseteq X$ and $Y \subsetneq A_j$. However, 
since $A$ is EFX, $u_i(A_i) \geq \max_{v \in A_j}u_i(A_j \setminus \{v\}) \geq u_i(Y)$. This contradicts with the fact that $u_i(A_i) <u_i(Y)$. 
\end{proof}

Note that the existence of EFX is not guaranteed even for two agents with identical utilities on a path; see Example A.5 in \citep{BiloCaFl22}.

\section{The case of $n$ agents}\label{sec:nagents}
In this section, we consider the case of $n$ agents. We start by showing that any MNW connected allocation satisfies $(1/2)$-PMMS.

\begin{theorem}\label{thm:MNW-half-PMMS}
For a connected graph $G$ and $n$ agents with additive utilities, any MNW connected allocation satisfies $(1/2)$-PMMS. 
\end{theorem}
\begin{proof}
For a connected graph $G$, let $A^*$ be any MNW connected allocation. Then take any pair of agents $i$ and $j$ such that $A^*_i=\emptyset$, or $i$ and $j$ are adjacent under $A^*$. Let $M_{ij} = A^*_i \cup A^*_j$. The allocation $A^*_{ij} = ( A^*_i, A^*_j )$ is an MNW allocation of $G[M_{ij}]$ to $i$ and $j$. We wish to show that $u_i(A^*_i) \ge \PMMS_i(M_{ij})/2$ and $u_j(A^*_j) \ge \PMMS_i(M_{ij})/2$. Note that if $u_i(A^*_i) = 0$ or $u_j(A^*_j) = 0$, then $j$ or $i$ must receive a bundle with utility $u_j(M_{ij})$ and $u_i(M_{ij})$, respectively. Thus, in either case, both agents receive a bundle with utility at least their PMMS.

Assume now that $u_i(A^*_i) > 0$ and $u_j(A^*_j) > 0$. Let $( A_1, A_2 )$ be a PMMS partition of $M_{ij}$ for agent $i$. Then at least one of the bundles has utility at least $u_j(M_{ij})/2$ according to $u_j$. Let $A = ( A_i, A_j )$ be the allocation in which $j$ receives this bundle and $i$ the other bundle. Then $u_i(A_i) \ge \PMMS_i(M_{ij})$. Moreover, it holds that $u_i(A^*_i) \cdot u_j(A^*_j)  \ge  u_i(A_i) \cdot u_j(A_j)$. 
Substituting $u_j(A^*_j)$ by an upper bound, and $u_i(A_i)$ and $u_j(A_j)$ by lower bounds, $u_i(A^*_i) \cdot u_j(M_{ij}) \ge \PMMS_i(M_{ij}) \cdot u_j(M_{ij})/2$, which implies $u_i(A^*_i) \ge \PMMS_i(M_{ij})/2$. 
Thus, agent $i$ receives in $A^*$ a bundle with utility at least a half of her PMMS. By exchanging $i$ and $j$ in the argument, $j$ must also receives in $A^*$ a bundle with utility at least a half of her PMMS.
\end{proof}

\begin{corollary}\label{cor:MNW-half-PMMS}
For a connected graph $G$ and $n$ agents with additive utilities, a Pareto-optimal and $(1/2)$-PMMS connected allocation exists. 
\end{corollary}

The one (1/2)-PMMS guarantee in \cref{thm:MNW-half-PMMS} is unfortunately the best that we can hope for using MNW allocations. Even if instances are restricted to two agents on a path with binary additive utilities, there exist instances with MNW allocations that provide one of the agents with exactly half their PMMS. Further, a leximin allocation does not satisfy any approximation of PMMS, even when agents have binary additive utilities on paths.

\begin{proposition}\label{prop:MNW}
Even for a path and two agents with binary additive utilities, there exists an instance with an MNW allocation in which one of the agents receives no more than a half of her PMMS.
\end{proposition}

\begin{proof}
Consider an instance in which a path of 6 items is to be allocated to two agents with binary additive utility functions. Assume that items are numbered from left to right on the path. Let $S_1 = \{1, 3, 4, 5\}$ and $S_2 = \{2, 3, 4, 6\}$ be the set of items with utility $1$ for agent 1 and 2, respectively. Then the MNW of the instance is $4$, and there are four MNW allocations: $( \{1, 2, 3\}, \{4, 5, 6\} )$, $( \{4, 5, 6\}, \{1, 2, 3\})$, $( \{1\}, \{2, 3, 4, 5, 6\} )$ and $( \{1, 2, 3, 4, 5\}, \{6\} )$. Note that each agent's PMMS is $2$, as they both have a total utility of $4$, and in the partition $( \{1, 2, 3\}, \{4, 5, 6\} )$ both bundles have utility $2$. However, in the two latter MNW allocations, one of the agents only receive a bundle with utility $1$. Thus, receiving only exactly a half of her PMMS.
\end{proof}

\begin{proposition}\label{prop:leximin}
    For any $\alpha > 0$, there exists an instance of a path and two agents with binary additive utilities such that no leximin connected allocation is $\alpha$-PMMS.
\end{proposition}

\begin{proof}
    For a given $\alpha > 0$, let $m = 2\ceil{2/\alpha} + 2$ be the number of items on the path. Let $S_1 = \{1, 2, \dots, m\}$ and $S_2 = \{3, m\}$ be the sets of items with utility $1$ for agent $1$ and $2$, respectively. Note that the PMMS of agent $1$ is $\floor{m / 2} = m / 2$.
    
    There exists a single leximin allocation for the instance. In this leximin allocation agent $1$ receives the bundle $\{1, 2\}$ and agent $2$ receives the bundle $\{3, 4, \dots, m\}$, providing each agent with a utility of 2. Any other allocation gives either agent $1$ or $2$ a bundle with utility less than 2. For the leximin allocation $A$, we get that
    \[
        u_1(A_1) = 2 = \frac{2}{\frac{m}{2}} \PMMS_1 = \frac{4}{m}\PMMS_1 \le \frac{4}{\frac{4}{\alpha} + 2} \PMMS_1 < \alpha \cdot \PMMS_1
    \]
    Thus, agent $1$ receives strictly less utility than $\alpha \cdot \PMMS_1$ and the allocation is not $\alpha$-PMMS.
\end{proof}

Below, we show that unlike PMMS, we cannot guarantee an $\alpha$-EF1 connected allocation, even among two agents for any $\alpha>0$. 


\begin{restatable}{proposition}{propAlphaEF}\label{prop:alphaEF1}
For any $\alpha \in (0,1]$, there exists an instance of a star and two agents with identical binary additive utilities such that no $\alpha$-EF1 connected allocation exists.
\end{restatable}

\begin{proof}
Consider any $\alpha \in (0,1]$. Choose a positive integer $\beta$ such that $\alpha > \frac{1}{\beta}$. Construct a star with $(\beta+2)$ leaves. Each of the two agents has utility one for every vertex of the star. Consider any connected allocation $(X,Y)$, where $X$ is the bundle containing the central vertex of the star. Then, we have $|X|=\beta+1$ and $|Y|=1$. Thus, $1=u(Y) < \frac{\beta +1}{\beta} \leq \alpha \cdot u(X \setminus \{v\})$ for any vertex $v \in X$. This means that $(X,Y)$ is not $\alpha$-EF1. 
\end{proof}

Note that \cref{thm:MNW-half-PMMS} does not provide us an efficient algorithm to compute a $(1/2)$-PMMS connected allocation. In fact, finding an MNW allocation is NP-hard even when $G$ is a tree and agents have binary additive utilities~\cite{Igarashi_Peters_2019}. In the next sections, we identify several cases in which an approximate/exact PMMS connected allocation can be efficiently computed.

\section{The case of two agents}\label{sec:twoagents}
For any connected graph $G$ and two agents with additive utilities, a PMMS connected allocation can always be constructed by the following cut-and-choose procedure: let one agent compute his or her PMMS partition and another agent choose a preferred bundle, leaving the remainder for the first agent.\footnote{Note that this method does not work for non-additive utilities. Indeed, even for two agents with submodular utilities on a cycle, there is an instance for which a PMMS connected allocation does not exist. See Proposition~\ref{prop:PMMS:two:nonadditive} in Appendix~\ref{appendix:twoagents}.} However, computing such a partition can easily be shown to be NP-hard by a simple reduction from {\sc PARTITION}. This leads us to the following question: can we obtain good approximation guarantees with respect to PMMS in polynomial time? We answer this question affirmatively for two agents, by showing that a $(3/4)$-PMMS connected allocation among two agents can always be computed in polynomial time for any additive utility functions. Note that for two agents, PMMS and MMS coincide with each other, and thus the $(3/4)$-PMMS guarantee is equivalent to $(3/4)$-MMS. 

\begin{theorem}\label{thm:halfPMMS-polynomial:twoagents}
For a connected graph $G$ and two agents with additive utility functions, a $(3/4)$-PMMS (equivalently, $(3/4)$-MMS) connected allocation can be found in polynomial time. 
\end{theorem}

To prove Theorem~\ref{thm:halfPMMS-polynomial:twoagents}, it suffices to show that a $(3/4)$-PMMS connected allocation among two agents with identical utility can be computed in polynomial time.

\begin{restatable}{theorem}{PMMStwoidentical}\label{thm:halfPMMS-polynomial:twoagents:identical}
    For a connected graph $G$ and two agents with identical additive utility function $u$, a $(3/4)$-PMMS (equivalently, $(3/4)$-MMS) connected allocation can be found in polynomial time.
\end{restatable}

We first observe that whenever $G$ is \emph{biconnected}, a $(3/4)$-PMMS connected allocation can be computed in polynomial time. Formally, we say that a vertex $v$ is a \emph{cut vertex} of a connected graph $G$ if removing $v$ makes $G$ disconnected. A graph $G$ is \emph{biconnected} if $G$ does not have a cut vertex. A \emph{bipolar ordering} of a graph $G$ is an enumeration $v_1, v_2, \dots, v_k$ of $G$'s vertices such that the subgraph induced by any initial or final segment of the enumeration is connected in $G$, i.e., both $\{v_1,\ldots,v_{j}\}$ and $\{v_{j +1},\ldots,v_{k}\}$ are connected in $G$ for every $j \in [k]$. It is known that $G$ admits a bipolar ordering between any pair of vertices if $G$ is biconnected.

\begin{lemma}[\cite{Lempel1967}]\label{lem:bipolar}
For a biconnected graph $G$ and any pair $v,w$ of vertices in $G$, a bipolar ordering over $G$ that starts with $v$ and ends with $w$ exists and can be computed in polynomial time. 
\end{lemma}

Consider Algorithm~\ref{alg:PMMS:two:biconnected}. For a biconnected graph $G$, the algorithm either chooses a valuable item and its complement to construct a $(3/4)$-PMMS connected allocation, or computes a bipolar ordering over the graph and applies the discrete cut-and-choose algorithm.\footnote{This algorithm is used to construct an EF1 connected allocation of $G$ possessing a bipolar ordering; see Proposition 3.4 of \cite{BiloCaFl22}.} Note that Line~\ref{line:bipolar} of the algorithm is well-defined since $u(\{v_1,\ldots,v_k\}) \geq u(\emptyset)$. 

\begin{algorithm}                      
\caption{Algorithm for constructing a $(3/4)$-PMMS connected allocation for two agents with identical additive utilities $u$ when $G=(V,E)$ is biconnected}  
\label{alg:PMMS:two:biconnected}                          
\begin{algorithmic}[1]  
\STATE Choose vertices $v^*,w^*$ with highest, $u(v^*)$, and second highest, $u(w^*)$, utility in $G$. Set $Y=\{v^*\}$ and $X=V \setminus Y$. \label{line:valuablevertex}
\IF{$u(Y) < \frac{3}{8} u(V)$}
\STATE Compute a bipolar ordering $(v_1,v_2,\ldots,v_{k})$ over $G$ with $v_1=v^*$ and $v_k=w^*$.\label{line:computebipolar}
\STATE Find smallest $j$ with $u(\{v_1,\ldots,v_j\}) \geq u(\{v_{j+1},\ldots,v_k\})$.\label{line:bipolar}
\IF{$u(\{v_1,\ldots,v_{j-1}\}) \geq u(\{v_{j+1},\ldots,v_k\})$}
\STATE Set $Y=\{v_1,\ldots,v_{j-1}\}$ and $X=\{v_{j},\ldots,v_k\}$. 
\ELSE 
\STATE Set $Y=\{v_1,\ldots,v_{j}\}$ and $X=\{v_{j+1},\ldots,v_k\}$. 
\ENDIF 
\ENDIF
\RETURN $(X,Y)$
\end{algorithmic}
\end{algorithm}

\begin{lemma}\label{lem:PMMS:two:biconnected}
Let $G$ be a biconnected graph. Suppose that there are two agents with identical additive utility function $u$. Then, Algorithm~\ref{alg:PMMS:two:biconnected} computes in polynomial time, a $(3/4)$-PMMS connected allocation $(X,Y)$ where 
\begin{itemize}
\item[$(${\rm i}$)$] $Y$ consists of a single vertex and $u(Y) \geq \frac{3}{8}u(V)$, 
\item[$(${\rm ii}$)$] $\min \{ u(X), u(Y) \} \geq \frac{3}{8} u(V)$, or 
\item[$(${\rm iii}$)$] $\min \{ u(X), u(Y) \} \geq \frac{1}{2}u(V \setminus \{z^*\})$, where $z^* \in \argmax \{\, u(v)\mid v \in V \setminus \{v^*,w^*\} \,\}$.
\end{itemize}
\end{lemma}
\begin{proof}
By \cref{lem:bipolar}, Algorithm~\ref{alg:PMMS:two:biconnected} clearly runs in polynomial time.
To see that it returns a $(3/4)$-PMMS connected allocation, let $(X,Y)$ be the resulting allocation of Algorithm~\ref{alg:PMMS:two:biconnected}. 

Suppose that there is a vertex $v$ with $u(\{v\}) \geq \frac{3}{8}u(V)$. Then by the choice of $v^*$ in Line~\ref{line:valuablevertex}, $u(\{v^*\}) \geq \frac{3}{8}u(V)$ and therefore the resulting allocation $(X,Y)$ satisfies $(${\rm i}$)$. This also means that $u(\{v^*\}) \geq \frac{3}{4}(\frac{1}{2}u(V)) \geq \frac{3}{4} \PMMS(V)$. Further, in any PMMS partition $(X',Y')$, at least one part is contained in $X=V \setminus \{v^*\}$, meaning that $u(X)$ is at least $\PMMS(V)$. Thus, $(X,Y)$ is $(3/4)$-PMMS. Since $G$ is biconnected, $X$ is connected in $G$ and hence $(X,Y)$ is a connected allocation. 

Suppose that $u(\{v\}) < \frac{3}{8}u(V)$ for each vertex $v \in V$. In this case, the algorithm computes a bipolar ordering in Line~\ref{line:computebipolar}. By definition, $(X,Y)$ is a connected allocation. 
Observe that there is at most one vertex $z^* \in V \setminus \{v^*,w^*\}$ with utility $u(z^*) > \frac{1}{4}u(V)$ since otherwise, there would be four vertices with utility greater than $\frac{1}{4}u(V)$. 

If there is no such vertex, by the fact that $\min \{u(X),u(Y)\} \geq \max \{u(X),u(Y)\} - u(v_j)$ and $u(V \setminus \{v_j\}) \geq \frac{3}{4}u(V)$, we obtain $\min \{u(X),u(Y)\} \geq \frac{1}{2} u(V \setminus \{v_j\})  \geq \frac{1}{2}(\frac{3}{4}u(V))$ and $(${\rm ii}$)$ holds, which means that $\min \{u(X),u(Y)\} \geq \frac{3}{4}(\frac{1}{2}u(V)) \geq \frac{3}{4} \PMMS(V)$. 

Now, consider the case when there is a vertex $z^* \in V \setminus \{v^*,w^*\}$ with utility $u(z^*) > \frac{1}{4}u(V)$. Then, the vertex $v_j$ computed in Line~\ref{line:bipolar} of the algorithm has utility at most $u(z^*)$. Indeed, since $u(v^*) < \frac{3}{8}u(V) < \frac{1}{2}u(V)$, we have $u(v_1) < u(V \setminus \{v_1\})$ and thus $j>1$. Also, since $u(w^*) < \frac{1}{2}u(V)$, we have $u(\{v_1,\ldots,v_{k-1}\}) > u(\{v_k\})$ and thus $j<k$. Thus, we obtain $\min \{u(X),u(Y)\} \geq \frac{1}{2} u(V \setminus \{v_j\}) \geq \frac{1}{2} u(V \setminus \{z^*\})$ and $(${\rm iii}$)$ holds. Further, we claim that $\PMMS(V) \leq u(V)- 2u(z^*)$. Indeed, in any PMMS partition, one bundle contains at least two vertices of $\{v^*,w^*,z^*\}$, each of which has utility at least $ u(z^*)$, implying that the other bundle has utility at most $u(V)- 2u(z^*)$. Thus, 
\begin{align*}
\min \{u(X),u(Y)\} 
&\geq \frac{1}{2} u(V \setminus \{z^*\})\\
&\geq \frac{1}{2} (u(V) -2 u(z^*)) + \frac{1}{2} u(z^*)\\
&\geq \frac{1}{2} \PMMS(V) + \frac{1}{8} u(V)\\
&\geq \frac{1}{2} \PMMS(V) + \frac{1}{4} \PMMS(V)\\
&\geq \frac{3}{4} \PMMS(V).  
\end{align*}
We conclude that $(X,Y)$ is a $(3/4)$-PMMS connected allocation. 
\end{proof}

For the case when $G$ is not necessarily biconnected, we cannot use the technique described in the proof of Lemma~\ref{lem:PMMS:two:biconnected}, since removing a single vertex of $G$ may disconnect the graph or $G$ may not admit a bipolar ordering (e.g., $G$ may be a star). Nevertheless, we can exploit the acyclic structure of the so-called \emph{block decomposition} and obtain sufficient conditions under which a PMMS connected allocation exists. Formally, a \emph{decomposition} of a graph $G=(V,E)$ is a family $\{F_1,F_2,\ldots,F_t\}$ of edge-disjoint subgraphs of $G$ such that $\bigcup^t_{i=1}E(F_i)=E$ where $E(F_i)$ is the set of edges of $F_i$. A \emph{block} of $G$ is a maximal biconnected subgraph of $G$. 

For a connected graph $G$, consider a bipartite graph $B(G)$ with bipartition $(\mathcal{B},S)$, where $\mathcal{B}$ is the set of blocks of $G$ and $S$ the set of cut vertices of $G$; a block $B$ and a cut vertex $v$ are adjacent in $B(G)$ if and only if $v$ belongs to $B$. Since every cycle of a graph is included in some block, the graph $B(G)$ is a tree:

\begin{lemma}[Prop.~5.3 in \cite{Bondy:2008}]\label{lem:blocktree}
The set of blocks forms a decomposition of a connected graph $G$ and the graph $B(G)$ is a tree. 
\end{lemma}

Now, observe that since $B(G)$ is a tree, removing an edge between an arbitrary cut vertex $c$ and its adjacent block $B$ in $B(G)$ results in two connected components $X'$ and $Y'$, where one part $X'$ contains $B$ and another part $Y'$ contains $c$. This partition of the block graph induces a partition $(X(B,c),Y(B,c))$ of the original vertices in $G$, where $X(B,c)$ is the set of vertices in $G$ that belong to $X'$ except for $c$ and $Y(B,c)$ is the set of vertices in $G$ that belong to $Y'$. 
See Figure \ref{fig:partition} for an illustration. 

\begin{figure*}
\centering
\begin{minipage}[b]{\columnwidth}
\centering
\begin{tikzpicture}[scale=0.6, transform shape,every node/.style={minimum size=6mm, inner sep=1pt}]

\draw [rounded corners,fill=gray!20] (2,0)--(4,0)--(3,-1)--cycle;

\draw[dotted,thick] (6,0.3) ellipse (2.6cm and 2cm);
\node at (6.3,1) {\Large $Y(B,c)$};

\draw[rotate=-45,dotted,thick] (1.4,0.6) ellipse (2.3cm and 1.6cm);
\node at (1,-0.45) {\Large $X(B,c)$};

\node at (2.9,-0.35) {\Large $B$};

\node[draw, circle,fill=white](1) at (0,0) {};
\node[draw, circle,fill=white](2) at (2,0) {};
\node[draw, circle,fill=white](3) at (4,0) {\Large $c$};
\node[draw, circle,fill=white](4) at (6,0) {};
\node[draw, circle,fill=white](5) at (8,0) {};

\node[draw, circle,fill=white](6) at (1,1) {};
\node[draw, circle,fill=white](7) at (4,1) {};
\node[draw, circle,fill=white](8) at (5,1) {};

\node[draw, circle,fill=white](9) at (1,-1) {};
\node[draw, circle,fill=white](10) at (3,-1) {};
\node[draw, circle,fill=white](11) at (5,-1) {};
\node[draw, circle,fill=white](12) at (7,-1) {};
\node[draw, circle,fill=white](13) at (2,-2) {};

\draw[-, >=latex,thick] (1)--(2) (2)--(3) (3)--(4) (4)--(5) (1)--(6) (2)--(6) (7)--(3) (3)--(8) (7)--(8);

\draw[-, >=latex,thick] (2)--(10) (10)--(3) (9)--(10) (9)--(13) (13)--(10) (3)--(11) (11)--(4) (12)--(4);
\end{tikzpicture}
\caption{A partition $(X(B,c),Y(B,c))$ of a graph.}
\label{fig:partition}
\end{minipage}\hspace{\columnsep}%
\begin{minipage}[b]{\columnwidth}
\centering
\begin{tikzpicture}[scale=0.6, transform shape,every node/.style={minimum size=6mm, inner sep=1pt}]

\draw [rounded corners,fill=gray!20] (2,0)--(4,0)--(3,-1)--cycle;
\node at (3,-0.35) {\Large $B^*$};

\draw[dotted,thick] (6,0.3) ellipse (2.6cm and 2cm);
\node at (6.3,1) {\large $Y(B^*,c)$};

\draw [rounded corners,dotted,thick] (-0.8,-0.3)--(2.8,-0.3)--(1,1.8)--cycle;
\node at (1.1,0.3) {\large $Y(B^*,a)$};

\draw [rounded corners,dotted,thick] (0.3,-0.7)--(3.7,-0.7)--(2,-2.7)--cycle;
\node at (2,-1.3) {\large $Y(B^*,b)$};

\node[draw, circle,fill=white](1) at (0,0) {};
\node[draw, circle,fill=white](2) at (2,0) {\Large $a$};
\node[draw, circle,fill=white](3) at (4,0) {\Large $c$};
\node[draw, circle,fill=white](4) at (6,0) {};
\node[draw, circle,fill=white](5) at (8,0) {};

\node[draw, circle,fill=white](6) at (1,1) {};
\node[draw, circle,fill=white](7) at (4,1) {};
\node[draw, circle,fill=white](8) at (5,1) {};

\node[draw, circle,fill=white](9) at (1,-1) {};
\node[draw, circle,fill=white](10) at (3,-1) {\Large $b$};
\node[draw, circle,fill=white](11) at (5,-1) {};
\node[draw, circle,fill=white](12) at (7,-1) {};
\node[draw, circle,fill=white](13) at (2,-2) {};

\draw[-, >=latex,thick] (1)--(2) (2)--(3) (3)--(4) (4)--(5) (1)--(6) (2)--(6) (7)--(3) (3)--(8) (7)--(8);

\draw[-, >=latex,thick] (2)--(10) (10)--(3) (9)--(10) (9)--(13) (13)--(10) (3)--(11) (11)--(4) (12)--(4);

\node[ultra thick] at (9.3,0) {\Huge $\Rightarrow$};

\begin{scope}[xshift=230,yshift=0]
    \node[draw, circle,fill=white](2') at (2,0) {\Large $a$};
    \node[draw, circle,fill=white](3') at (4,0) {\Large $c$};
    \node[draw, circle,fill=white](10') at (3,-1) {\Large $b$};
    \draw[-, >=latex,thick]  (2')--(3') (3')--(10') (2')--(10');
\node at (3,-0.35) {\Large $G'$};
		        \end{scope}
\end{tikzpicture}
\caption{Merge operation applied to a block $B^*$.}
\label{fig:merge}
\Description{Shows the merge operation applied to the block B* = \{a, b, c\} in the graph G = (\{0, 1, 2, 3, 4, 5, 6, 7, 8, 9, a, b, c\}, \{\{0, 1\}, \{0, a\}, \{1, a\}, \{2, 3\}, \{2, b\}, \{3, b\}, \{a, b\}, \{b, c\}, \{a, c\}, \{c, 4\}, \{4, 5\}, \{c, 5\}, \{c, 6\}, \{6, 7\}, \{c, 7\}, \{7, 8\}, \{7, 9\}\}). Highlights Y(B*, a) = \{0, 1, a\}, Y(B*, b) = \{2, 3, b\} and Y(B*, c) = \{4, 5, 6, 7, 8, 9, c\}. Finally, shows the resulting graph G' = (\{a, b, c\}, \{\{a, b\}, \{b, c\}, \{c, a\}\} from the merge operator.}
\end{minipage}%
\end{figure*}

It turns out that if no ``local" improvement is possible, an allocation of form $(X(B,c),Y(B,c))$ is in fact a PMMS connected allocation with respect to identical utility function $u$.

\begin{restatable}{lemma}{lemPMMStwoidentical}\label{lem:PMMS:two}
Let $G$ be a connected graph. Suppose that there are two agents with identical additive utility function $u$. Let $(B,c)$ be a pair of a block $B$ and a cut vertex $c$ included in $B$ such that $u(X(B,c)) \leq u(Y(B,c)) $ and $u(X(B,c)) \geq \min \{u(X(B',c)),u(Y(B',c))\}$ for every block $B'$ containing $c$. Then, the partition $(X(B,c)),Y(B,c))$ is a PMMS connected allocation. 
\end{restatable}

\begin{proof}
Consider any PMMS partition $(X,Y)$ with $X \neq X(B,c)$ and $Y \neq Y(B,c)$. Without loss of generality, assume that $Y \cap Y(B,c) \neq \emptyset$. Since any pair of the blocks adjacent to $c$ is connected through $c$ only, either $X$ or $Y$ includes at least one of $X(B,c)$ and $Y(B,c)$. If $Y(B,c) \subsetneq Y$, then we have $X \subseteq X(B,c)$ and thus $u(X) \leq u(X(B,c)) \leq u(Y(B,c)) \leq u(Y)$, which implies that $(X(B,c),Y(B,c))$ is a PMMS partition. Hence, assume $X(B,c) \subsetneq X$. Then we should have $c \in X$, meaning that $X$ contains all the blocks containing $c$ except for $B$ and $Y$ intersects with at most one block adjacent to $c$. By the assumption that 
$$
u(X(B,c)) \leq u(Y(B,c))~\mbox{and}~u(X(B,c)) \geq \min \{u(X(B',c)),u(Y(B',c))\}
$$
for every block $B'$ containing $c$, we have $u(X(B,c)) \geq u(Y)$. Thus, together with the fact that $u(X) \geq u(X(B,c))$, we obtain 
\[
u(X(B,c)) \geq \min \{u(X),u(Y)\}=u(Y),
\]
implying that $(X(B,c),Y(B,c))$ is a PMMS partition. 
\end{proof}


Now, we are ready to prove Theorem~\ref{thm:halfPMMS-polynomial:twoagents:identical}. 

\begin{proof}
Let $G$ be a connected graph and $u$ be an identical additive utility function. 
If $G$ is biconnected, by Lemma~\ref{lem:PMMS:two:biconnected}, a $(3/4)$-PMMS connected allocation can be found in polynomial time.
Suppose that $G$ is not biconnected. Let $(B^*,c^*)$ be a pair of block and cut vertex where $(X(B^*,c^*),Y(B^*,c^*))$ maximizes the minimum of $u(X(B,c))$ and $u(Y(B,c))$ over all pairs of cut vertices $c$ and its adjacent block $B$. If we have $u(X(B^*,c^*)) \leq u(Y(B^*,c^*))$, then $(X(B^*,c^*),Y(B^*,c^*))$ is a PMMS partition by Lemma~\ref{lem:PMMS:two} and by the choice of $(B^*,c^*)$. 

Thus, assume that $u(X(B^*,c^*)) \geq u(Y(B^*,c^*))$. Note that for each cut vertex $c$ of $B^*$ with $c \neq c^*$, we have $Y(B^*,c*) \subseteq X(B^*,c)$, and further 
\[
\min \{u(X(B^*,c)),u(Y(B^*,c)) \}  \leq \min \{u(X(B^*,c)^*),u(Y(B^*,c^*))\} 
\]
which implies $u(Y(B^*,c)) \leq u(Y(B^*,c^*))$. 
For each cut vertex $c$ adjacent to $B^*$, merge $Y(B^*,c)$ into $c$, namely, we replace the vertices in $Y(B^*,c)$ with a single vertex $c$ and there is an edge between $c$ and another vertex $w$ in the new graph whenever $w$ is adjacent to $c$ in the original graph (see \Cref{fig:merge} for an illustration). Let $G'$ denote the resulting graph. It is easy to see that $G'$ is biconnected. Moreover, we define a new additive utility function $u'$ on $G'$ where the utility of an agent for each vertex $v$ of $G'$ is equal to her utility for all the vertices of $G$ that are merged into $v$. Note that for any connected allocation $(X,Y)$ of $G'$, it is not difficult to see that the corresponding allocation $(X^*,Y^*)$ of the original vertices in $G$ is connected. 

Apply Algorithm~\ref{alg:PMMS:two:biconnected} for $G'$ and $u'$. We will show that the resulting allocation corresponds to a $(3/4)$-PMMS connected allocation of the original vertices. Let $(X,Y)$ denote the output of Algorithm~\ref{alg:PMMS:two:biconnected} for $G'$ and $u'$. 
Let us denote by $X^*$ (resp. $Y^*$) the set of vertices of $G$ merged into some vertex in $X$ (resp. $Y$). As discussed before, $(X^*,Y^*)$ is a connected allocation in $G$. 

Suppose that $(X,Y)$ satisfies the condition~$(${\rm i}$)$~of Lemma~\ref{lem:PMMS:two:biconnected}. If $Y^*$ consists of a single vertex in $G$, by a similar argument to the proof of Lemma~\ref{lem:PMMS:two:biconnected}, it can be easily verified that $(X^*,Y^*)$ satisfies $(3/4)$-PMMS. Indeed, since the total utility under $G'$ is the same as that under $G$, we have $u(Y) \geq \frac{3}{8}u(V)$, which means $u(Y) \geq \frac{3}{4}(\frac{1}{2}u(V)) \geq \frac{3}{4} \PMMS(V)$. Further, in any PMMS partition $(X',Y')$ of $G$,  one part $X'$ is contained in $X^*$, meaning that $u(X') \leq u(X^*)=u(V \setminus Y )$ and the utility under $X^*$ is at least $\PMMS(V)$. Thus, $(X,Y)$ is $(3/4)$-PMMS. 

If $Y^*=Y(B^*,c)$ for some cut vertex $c$ that belongs to $B^*$, then by the choice of $(B^*,c^*)$ and $v^*$ in Algorithm~\ref{alg:PMMS:two:biconnected}, we have that $u(Y^*)=u(Y(B^*,c))=u(Y(B^*,c^*))$,  
and by additivity, $u(X^*)=u(X(B^*,c^*))$. Thus, we have $u(Y^*) \geq \frac{3}{8}u(V)$ and $u(Y^*) \leq u(X^*)$, where the latter inequality holds as a result of the assumption that $u(X(B^*,c^*)) \geq u(Y(B^*,c^*))$. Thus, we have $u(X^*) \geq \frac{1}{2}u(V) \geq \PMMS(V)$ and $u(Y^*) \geq \frac{3}{4}\PMMS(V)$, concluding that $(X^*,Y^*)$ satisfies $(3/4)$-PMMS. 

Suppose that $(X,Y)$ satisfies the condition~$(${\rm ii}$)$~of Lemma~\ref{lem:PMMS:two:biconnected}. Then since $u'(X)=u(X^*)$ and $u'(Y)=u(Y^*)$, we have that 
$$
\min \{ u(X^*), u(Y^*) \} \geq \frac{3}{8} u(V).
$$
Similar to the above, $(X^*,Y^*)$ satisfies $(3/4)$-PMMS. 

Suppose that $(X,Y)$ satisfies the condition~$(${\rm iii}$)$~of Lemma~\ref{lem:PMMS:two:biconnected}. 
Let $v^*,w^*,z^*$ be the vertices with highest, second highest, third highest utility in $G'$. For any connected allocation $(X',Y')$ of the original graph $G$, one bundle, say $X'$, must contain at least two original vertices that correspond to $v^*,w^*$, or $z^*$. Let $V^*$, $W^*$, and $Z^*$ denote the set of original vertices in $G$ merged into $v^*,w^*$, and $z^*$, respectively. If $X'$ contains two vertices from $v^*,w^*,z^*$ and $Y'$ contains one vertex from $v^*,w^*,z^*$, then at least two sets of the merged vertices $V^*$, $W^*$, and $Z^*$ are fully contained in $X'$. Similarly, if $X'$ contains all the three vertices $v^*,w^*,z^*$, then $Y'$ is fully contained in one of the sets $V^*$, $W^*$, and $Z^*$. Thus, at least one bundle of $X'$ and $Y'$ has utility less than $u(V) - 2 u(Z^*)$, which means that $\PMMS(V) \leq u(V) -2 u(Z^*)$. Thus, 
\begin{align*}
\min \{u(X^*),u(Y^*)\} 
&\geq \frac{1}{2} u(V \setminus Z^*)\\
&\geq \frac{1}{2} (u(V) -2 u(Z^*)) + \frac{1}{2} u(Z^*)\\
&\geq \frac{1}{2} \PMMS(V) + \frac{1}{2} u(Z^*)\\
&\geq \frac{1}{2} \PMMS(V) + \frac{1}{8} u(V)\\
&\geq \frac{1}{2} \PMMS(V) + \frac{1}{4} \PMMS(V)\\
&\geq \frac{3}{4} \PMMS(V).  
\end{align*}
We conclude that $(X^*,Y^*)$ is a $(3/4)$-PMMS connected allocation. 

It remains to analyze the running time of the above procedure. Note that one can compute the block graph $B(G)$ of a connected graph $G$ in polynomial time~\cite{Hopcroft1973}. Further, there are at most $|V|^2$ pairs of blocks and cut vertices. Hence, we can compute $(B^*,c^*)$ in polynomial time. The remaining procedure clearly runs in polynomial time by Lemma~\ref{lem:PMMS:two:biconnected}. 
\end{proof}

\begin{proof}[Proof of Theorem \ref{thm:halfPMMS-polynomial:twoagents}]
Suppose that there are two agents with utility functions $u_1,u_2$. By Theorem~\ref{thm:halfPMMS-polynomial:twoagents:identical}, we can compute a $(3/4)$-PMMS connected allocation $(X,Y)$ with respect to $u_1$ in polynomial time. Then, the allocation that assigns to agent $2$ a preferred bundle among $X$ and $Y$ and the rest to agent $1$ is $(3/4)$-PMMS. 
\end{proof}

\section{The case of three agents}\label{sec:paths}

In this section, we prove that a PMMS connected allocation always exists and can be found in polynomial time for three agents with additive utilities on a path. While limited, this case is interesting as the existence of PMMS remains open for three agents with additive utilities in the standard setting.


\begin{restatable}{theorem}{theoremThreeAgentsPath}\label{thr:three-agents-path}
For a path and three agents with additive utilities, a PMMS connected allocation always exists and can be found in polynomial time.
\end{restatable}


The proof of \cref{thr:three-agents-path} relies on case analysis. Before proceeding with the detailed proof, we briefly outline the method used. First a PMMS partition is computed for each of the three agents. This can be done in polynomial time by \cref{lem:SMMS-additive-trees-polynomial}. Each of the three PMMS partitions can be represented by a pair of edges, namely the edge separating the left and middle bundle and the edge separating the middle and right bundle. \Cref{fig:path-partitions-examples} shows three examples of the distribution of these edges. In each example, the three bars indicate the PMMS partitions of the three agents.


In most cases, such as in \cref{fig:path-partitions-examples-a}, a PMMS allocation can be found by simply allocating the bundles in one of the PMMS partitions to the agents in a certain way. For example, in \cref{fig:path-partitions-examples-a} a PMMS allocation can be obtained by giving the left and right bundles in the topmost PMMS partition (yellow) to the two bottom agents (blue and red) in any way and the middle bundle to the top agent (yellow). This is guaranteed to be a PMMS allocation, as the middle bundle of the topmost PMMS partition is contained within the middle bundle of each of the two other agents, and thus each of the two outer bundles in the topmost PMMS partition satisfy PMMS for the two bottom agents.




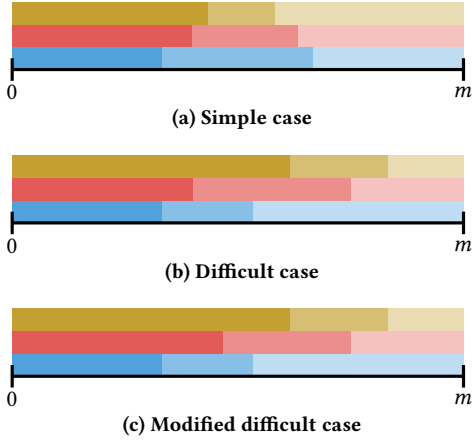
\begin{figure}[t]
    \begin{subfigure}{\columnwidth}
        \centering
        \begin{tikzpicture}
            \filldraw[fill=sblue!80, draw=none] (0, 0) rectangle (2, 0.3);
            \filldraw[fill=sblue!55, draw=none] (2, 0) rectangle (4, 0.3);
            \filldraw[fill=sblue!30, draw=none] (4, 0) rectangle (6, 0.3);
            
            \filldraw[fill=sred!80, draw=none] (0, 0.3) rectangle (2.4, 0.6);
            \filldraw[fill=sred!55, draw=none] (2.4, 0.3) rectangle (3.8, 0.6);
            \filldraw[fill=sred!30, draw=none] (3.8, 0.3) rectangle (6, 0.6);
    
            \filldraw[fill=syellow!80, draw=none] (0, 0.6) rectangle (2.6, 0.9);
            \filldraw[fill=syellow!55, draw=none] (2.6, 0.6) rectangle (3.5, 0.9);
            \filldraw[fill=syellow!30, draw=none] (3.5, 0.6) rectangle (6, 0.9);
            
            \draw[very thick] 
                (0, 0) -- (6, 0)
                (0, -0.15) -- (0, 0.15)
                (6, -0.15) -- (6, 0.15)
                ;
    
            \node at (0, -0.3) {$0$};
            \node at (6, -0.3) {$m$};
        \end{tikzpicture}
        \vspace*{-0.6\baselineskip}
        \caption{Simple case}
        \label{fig:path-partitions-examples-a}
    \end{subfigure}\\%
    \vspace{0.7\baselineskip}%
    \begin{subfigure}{\columnwidth}
        \centering
        \begin{tikzpicture}
            \filldraw[fill=sblue!80, draw=none] (0, 0) rectangle (2, 0.3);
            \filldraw[fill=sblue!55, draw=none] (2, 0) rectangle (3.2, 0.3);
            \filldraw[fill=sblue!30, draw=none] (3.2, 0) rectangle (6, 0.3);
            
            \filldraw[fill=sred!80, draw=none] (0, 0.3) rectangle (2.4, 0.6);
            \filldraw[fill=sred!55, draw=none] (2.4, 0.3) rectangle (4.5, 0.6);
            \filldraw[fill=sred!30, draw=none] (4.5, 0.3) rectangle (6, 0.6);
        
            \filldraw[fill=syellow!80, draw=none] (0, 0.6) rectangle (3.7, 0.9);
            \filldraw[fill=syellow!55, draw=none] (3.7, 0.6) rectangle (5, 0.9);
            \filldraw[fill=syellow!30, draw=none] (5, 0.6) rectangle (6, 0.9);
            
            \draw[very thick] 
                (0, 0) -- (6, 0)
                (0, -0.15) -- (0, 0.15)
                (6, -0.15) -- (6, 0.15)
                ;
        
            \node at (0, -0.3) {$0$};
            \node at (6, -0.3) {$m$};
            
        \end{tikzpicture}
        \vspace*{-0.6\baselineskip}
        \caption{Difficult case}
        \label{fig:path-partitions-examples-b}
    \end{subfigure}\\%
    \vspace{0.7\baselineskip}%
    \begin{subfigure}{\columnwidth}
        \centering
        \begin{tikzpicture}
            \filldraw[fill=sblue!80, draw=none] (0, 0) rectangle (2, 0.3);
            \filldraw[fill=sblue!55, draw=none] (2, 0) rectangle (3.2, 0.3);
            \filldraw[fill=sblue!30, draw=none] (3.2, 0) rectangle (6, 0.3);
            
            \filldraw[fill=sred!80, draw=none] (0, 0.3) rectangle (2.8, 0.6);
            \filldraw[fill=sred!55, draw=none] (2.8, 0.3) rectangle (4.5, 0.6);
            \filldraw[fill=sred!30, draw=none] (4.5, 0.3) rectangle (6, 0.6);
        
            \filldraw[fill=syellow!80, draw=none] (0, 0.6) rectangle (3.7, 0.9);
            \filldraw[fill=syellow!55, draw=none] (3.7, 0.6) rectangle (5, 0.9);
            \filldraw[fill=syellow!30, draw=none] (5, 0.6) rectangle (6, 0.9);
            
            \draw[very thick] 
                (0, 0) -- (6, 0)
                (0, -0.15) -- (0, 0.15)
                (6, -0.15) -- (6, 0.15)
                ;
        
            \node at (0, -0.3) {$0$};
            \node at (6, -0.3) {$m$};
        \end{tikzpicture}
        \vspace*{-0.6\baselineskip}
        \caption{Modified difficult case}
        \label{fig:path-partitions-examples-c}
    \end{subfigure}
    
    \caption{Three possible layouts of PMMS partitions for 3 agents on a path.}
    \label{fig:path-partitions-examples}
    \Description{Contains three sub-figures. Each sub-figure shows the path and a PMMS partition of each of the three agents. Let a1 (a2) be the edge connecting the first (second) and second (third) bundle in the PMMS partition of agent 1. Let b1 and b2, and c1 and c2 be the same for agent 2 and 3, respectively. Sub-figure (a) shows a case in which a1 < b1 < c1 < c2 < b2 < a2. Sub-figure (b) shows a case in which a1 < b1 < a2 < c1 < b2 < c2. This case is repeated in sub-figure (c). The only difference to sub-figure (b) is that b1 has been changed to b1' with b1 < b1' < a2.}
\end{figure}

In some cases, such as in \cref{fig:path-partitions-examples-b}, it may be that in every one of the three PMMS partitions, there is only a single bundle that satisfies PMMS for each of the two other agents. If this is the same bundle for both agents, another method must be used to construct a PMMS allocation. In these few cases, it can be shown that one of the PMMS partitions is such that the middle bundle has the highest utility for both of the other agents. For example, in the case in \cref{fig:path-partitions-examples-b}, this would be the case for the middle bundle of the middle agent. If this was not the case for the bottom (top) agent, then since the left (right) bundle covers more than an entire bundle in the agent's own PMMS partition, we can show that the left (right) bundle in the middle PMMS partition satisfies PMMS for the agent.

If both other agents think the middle bundle has the most utility, we can show that we can select one of the two and let her move the border between two of the bundles in the partition in a way that shrinks the middle bundle. Since the middle bundle had most utility, we can guarantee that the new partition will satisfy PMMS for the agent if she receives either of the modified bundles. \Cref{fig:path-partitions-examples-c} shows an example of such a modification to the case in \cref{fig:path-partitions-examples-b}.

We now proceed with the detailed proof of \cref{thr:three-agents-path}. First, we introduce some additional notation for use in the proof. Given a path $P$ of $m$ vertices, number the edges from left to right by the numbers 1, 2, $\dots$, $m - 1$. Let $[x, y]$ denote the bundle consisting of all items between edge $x$ and $y$, where the numbers $0$ and $m$ represent respectively the start and end of the path. For some $n > 0$, we say that the $(n - 1)$-tuple $(c_1, c_2, \dots, c_{n - 1})$ with $c_i \le c_{i + 1}$, $0 \le c_i \le m$ causes the partition $([0, c_1], [c_1, c_2], \dots, [c_{n - 1}, m])$.
We say that a bundle $B$ is \emph{good for utility function} $u$ in a partition if $u(B) \ge \PMMS(B \cup B')$ for every neighbouring bundle $B'$ in the partition.
That is, the requirements for PMMS is satisfied in the partition for any agent with utility function $u$, if she receives the bundle $B$. We say that a bundle is \emph{good for agent $i$} if the bundle is good for her utility function $u_i$. Any allocation in which all agents receive bundles that are good for them is a PMMS allocation.


First, we need to establish several properties of PMMS on paths under additive utilities. Some of these properties hold for more general graph classes or under monotone utilities. We begin by considering sufficient conditions for existence of a PMMS connected allocation on a path with three agents.

\begin{lemma}\label{lem:tree-monotone-better-than-neighbours}
    For a tree, a monotone utility function $u$, and a partition of the tree into connected bundles, let $B_i$ and $B_j$ be two neighbouring bundles such that $u(B_i) \ge u(B_j)$. Then, $u(B_i) \ge \PMMS(B_i \cup B_j)$.
\end{lemma}

\begin{proof}
    Let $(X, Y) \in \Pi_2(B_i \cup B_j)$ be any PMMS partition of the tree $B_i \cup B_j$. Since $X$ and $Y$ must be connected, it holds that either $X$ or $Y$ is a subset of either $B_i$ or $B_j$. Assume w.l.o.g.\ that this holds for $X$. Then $u(X) \le u(B_i)$ or $u(X) \le u(B_j) \le u(B_i)$ and we get that $\mathrm{PMMS}(B_i \cup B_j) \le u(X) \le u(B_i)$.
\end{proof}

\begin{lemma}\label{lem:tree-monotone-at-least-one}
    For a tree, a monotone utility function $u$ and a partition of the graph, there exists at least one good bundle for $u$.
\end{lemma}

\begin{proof}
    In any partition $(B_1, B_2, \dots, B_n)$ there is at least one bundle $B_i$ with $u(B_i) \ge u(B_j)$ for all $j \in \{1, \dots, n\}$. Applying \cref{lem:tree-monotone-better-than-neighbours} to each neighbour $B_j$ of $B_i$ guarantees that $B_i$ must be good for $u$.
\end{proof}

Given \cref{lem:tree-monotone-at-least-one}, we get the following sufficient conditions for the existence of a PMMS connected allocation. These are similar to the ones used by \citet{Feige:22C} when considering MMS existence for three agents.

\begin{lemma}\label{lem:3-agents-sufficient-conditions}
    For a path and three agents with monotone utility functions, let $(c_1, c_2)$ be a partition of the path, let $G_i$ be the set of bundles in the partition that are good for agent $i$. Then, any one of the following conditions is sufficient for a PMMS connected allocation to exist
    \begin{enumerate}[label=(\arabic*)]
        \item There are two agents $i, j$ such that $|G_i| = 2$, $|G_j| = 2$ and $|G_i \cap G_j| = 1$.\label{item:lem:3-agents-sufficient-conditions:1}
        \item There are two agents $i, j$ such that $|G_i| = 3$ and $|G_j| \ge 2$.\label{item:lem:3-agents-sufficient-conditions:2}
        \item There is some agent $i$ with $|G_i| = 3$, and for the two remaining agents $j, k$ we have that $G_j \setminus G_k \neq \emptyset$.\label{item:lem:3-agents-sufficient-conditions:3}
    \end{enumerate}
\end{lemma}

\begin{proof}
    For \ref{item:lem:3-agents-sufficient-conditions:1} let $B_k$ be a bundle that is good for the last agent, $k$. Give this bundle to $k$. Then for at least one of $i$ and $j$, $B_k$ is good. If both $B_k \in G_i$ and $B_k \in G_j$, then $G_i \cap G_j = \{B_k\}$ and a PMMS allocation is obtained by giving the bundle in $G_i \setminus \{B_k\}$ to $i$ and the bundle in $G_j \setminus \{B_k\}$ to $j$. Otherwise, assume w.l.o.g.\ that $B_k \in G_i$. Then, allocate $B_i$, the only bundle in $G_i \setminus \{B_k\}$, to agent $i$. Since $B_k \notin G_j$, there is some bundle $B_j \in (G_j \setminus \{B_i, B_k\})$ which can be allocated to agent $j$. Thus, each agent receives a bundle that is good for them and we have a PMMS connected allocation.

    For \ref{item:lem:3-agents-sufficient-conditions:2} and \ref{item:lem:3-agents-sufficient-conditions:3}, a PMMS connected allocation can be obtained by allocating some $B_k \in G_k$ to agent $k$. Then, there must be some remaining bundle $B_j \in (G_j \setminus \{B_k\})$ that can be given to agent $j$. The last bundle $B_i$ can be given to agent $i$, as every bundle is good for $i$.
\end{proof}

We now have sufficient conditions for obtaining PMMS connected allocations in the different cases that need to be considered. These cover most of the cases that will be considered in \cref{thr:three-agents-path} directly. In order to handle the remaining cases, we need the following properties of PMMS partitions.

\begin{lemma}\label{lem:move-value-p-mms-split}
    For a path and a monotone utility function $u$, let $x, y, z$ be such that $x \le y \le z$ and $a$ such that $y \le a \le z$. If $([x, y], [y, z])$ is a PMMS partition of $[x, z]$ for $u$ and $u([x, a]) > u([x, y])$, then $u([a, z]) \le u([x, a])$.
\end{lemma}

\begin{proof}
    Assume to the contrary that $u([a, z]) > u([x, a])$. Then $u([a, z]) > u([x, a]) > u([x, y])$. Thus, 
    \begin{align*}
        \min\{u([x, a]), u([a, z])\} &= u([x, a]) > u([x, y]) \\
        &\ge \min\{u([x, y]), u([y, z])\}
    \end{align*}
    This is a contradiction as $([x, y], [y, z])$ is by definition then not a PMMS partition of $[x, z]$.
\end{proof}

The following lemma is a stricter version of \cref{lem:move-value-p-mms-split} when utilities are additive.

\begin{lemma}\label{lem:conditions-p-mms-split}
    For a path and an additive utility function $u$, let $x, y, z$ be such that $x < y < z$. If and only if the partition $([x, y]$, $[y, z])$ is a PMMS partition of $[x, z]$ for $u$ it holds that
    \begin{enumerate}[label=(\arabic*)]
        \item $u([x, y]) = 0$ or $u([x, a]) \le u([y, z])$, where $a$ is the largest $a \ge x$ with $u([x, a]) < u([x, y])$
        \item $u([y, z]) = 0$ or $u([b, z]) \le u([x, y])$, where $b$ is the smallest $b \le z$ with $u([b, z]) < u([y, z])$
    \end{enumerate}
\end{lemma}

\begin{proof}
    We prove both directions. 
    
    \paragraph{$\implies$:}
    Assume that the partition $([x, y], [y, z])$ is a PMMS partition for $u$. If $u([x, y]) = 0$, then (1) holds. Now assume that $u([x, y]) \neq 0$. Since $u([x, x]) = u(\emptyset) = 0$, $u([x, y]) > 0$ and $u$ is additive, there must exists some $x < c \le y$ such that $u([x, c]) = u([x, y])$ and $u([x, c - 1]) < u([x, c])$. Then it follows that $a = c - 1$. If $u([x, a]) > u([y, z])$, the partition $([x, a]$, $[a, z])$ has the property that $$\min\{u([x, a]), u([a, z])\} > u([y, z]) \ge \min\{u([x, y]), u([y, z])\}\;.$$ A contradiction, as $([x, y], [y, z])$ is a PMMS partition. So it must hold that $u([x, a]) \le u([a, z])$. Now assume that $u([y, z]) \neq 0$. By the same argument, we get that $b$ exists and that 
    $$\min\{u([x, b]), u([b, z])\} > u([x, y]) \ge \min\{u([x, y]), u([y, z])\}\;.$$
    This is another contradiction as $([x, y], [y, z])$ is a PMMS partition for $u$.

    \paragraph{$\impliedby$:}
    If both $u([x, y]) = 0$ and $u([y, z]) = 0$, then since $u$ is additive, $u([x, z]) = 0$. Therefore $\PMMS([x, z]) = 0$ and any partition is a PMMS partition. Assume now that $u([x, y]) = 0$ and $u([b, z]) \le u([x, y])$. Then $u([b, z]) = 0$ and there is only a single item in $[x, z]$ with a non-zero utility. Thus, $\PMMS([x, z]) = 0$ and any partition is a PMMS partition. Similarly, if $u([y, z]) = 0$ and $u([x, a]) \le u([y, z])$, then $u([x, a]) = 0$ and there is only a single item in $[x, z]$ with non-zero utility. Thus, $\PMMS([x, z]) = 0$ and any partition is a PMMS partition.
    Finally, assume that $u([x, y]) \neq 0$, $u([y, z]) \neq 0$, $u([x, a]) \le u([y, z])$ and $u([b, z]) \le u([x, y])$. Then, there can not be any partition $([x, c]$, $[c, z])$ with $\min\{u([x, c]), u([c, z])\} > \min\{u([x, y]), u([y, z])\}$. Any such partition would have to increase the utility of either the left or the right bundle, whichever has the least utility, in which case it follows by the additive utilities that $u([x, c]) \le u([x, a])$ or $u([c, z]) \le u([b, z])$. Thus,
    \begin{align*}
    \min\{u([x, c]), u([c, z])\} 
    &\le \max\{u([x, a]), u([b, z])\} \\
    &\le \min\{u([x, y]), u([y, z])\}
    \end{align*}
\end{proof}

\begin{corollary}\label{cor:conditions-p-mms-split}
    For a path and an additive utility function $u$, let $x, y, z$ be such that $x < y < z$. If the partition $([x, y]$, $[y, z])$ is a PMMS partition for $u$ it holds that
    \begin{enumerate}[label=(\arabic*)]
        \item If $u([x, y]) \neq 0$, then $u([x, a]) < u([a, z])$, where $a$ is the largest $a \ge x$ with $u([x, a]) < u([x, y])$
        \item If $u([y, z]) \neq 0$, then $u([b, z]) < u([x, b])$, where $b$ is the smallest $b \le z$ with $u([b, z]) < u([y, z])$
    \end{enumerate}
\end{corollary}

Next we consider how the cut point for a PMMS partition of a subpath to two agents changes as we expand the subpath in one direction and shrink it in the opposite end. More specifically, we simply establish that if we move the start or/and end of the subpath left, while maintaining the same cut point or moving the cut point rightwards, then the left bundle will remain good for the utility function in the split of the subpath. While not very useful for bundles in the middle of the path, it is useful when considering if the outermost bundles on a path are good for an agent.

\begin{lemma}\label{lem:sub-and-super-paths}
    For a path and an additive utility function $u$, let $([x, y], [y, z])$ be a PMMS partition of $[x, z]$ for $u$, where $x \le y \le z$. Let $a$, $b$ and $c$ be such that $a \le x$, $b \ge y$ and $b < c \le z$. Then, $u([a, b]) \ge \PMMS([a, c])$.
\end{lemma}

\begin{proof}
    Since $u$ is monotone, it holds that $u([a, b]) \ge u([x, y])$ and $u([b, c]) \le u([y, z])$. Assume that $u([a, b]) < \PMMS([a, c])$. Then, there exists $d$ with $b < d < c$ such that $$\min(u([a, d]), u([d, c])) > u([a, b]) \ge u([x, y])\;.$$ Since $u$ is additive, $u([b, d]) > 0$ and by \cref{lem:conditions-p-mms-split} it must hold that $u([x, y]) \ge u([d, z]) \ge u([d, c])$. This is a contradiction, as it requires $u([x, y]) > u([x, y])$ and no such $d$ can exist.
\end{proof}

Note that \cref{lem:sub-and-super-paths} does not hold for general monotone utilities, as is highlighted in the following example.

\begin{example}
    Consider a path with five items and the monotone utility function
    \[
    u(S) =
    \begin{cases}
        1 & \text{if } [0, 3] \subseteq S \text{ or } [3, 5] \subseteq S \\
        0 & \text{otherwise}
    \end{cases}
    \]
    Then, $([1, 2], [2, 5])$ is a PMMS partition of the subpath $[1, 5]$ for $u$. However, for the left bundle to be good for $u$ on the (sub)path $[0, 5]$, it must at least contain $[0, 3]$. This is a contradiction to \cref{lem:sub-and-super-paths}, which states that the bundle $[0, 2]$ should be good for $u$ in $[0, 5]$. Thus, \cref{lem:sub-and-super-paths} does not hold for monotone utility functions.
\end{example}

The next result extends the result of \cref{lem:sub-and-super-paths} to the 3-agent case.

\begin{lemma}\label{lem:x-y-three-agents}
    For a path, three agents and an additive utility function $u$, let $(c_1, c_2)$ cause a PMMS partition for $u$ and $(x, y)$ be such that $x \ge c_1$ and $y \ge c_2$. Then either \ref{item:lem:x-y-three-agents:1} holds or \ref{item:lem:x-y-three-agents:2}, \ref{item:lem:x-y-three-agents:3} and \ref{item:lem:x-y-three-agents:4} hold.
    \begin{enumerate}[label=(\arabic*)]
        \item In the partition caused by $(x, y)$ the bundle $[0, x]$ is good for $u$.\label{item:lem:x-y-three-agents:1}
        \item In the partition caused by $(x, y)$, the bundle $[x, y]$ is good for $u$.\label{item:lem:x-y-three-agents:2}
        \item $u([0, x]) < u([x, y])$ and $u([y, m]) < u([x, y])$.\label{item:lem:x-y-three-agents:3}
        \item There exists $z$ with $x < z \le y$ such that at least one of the following holds\label{item:lem:x-y-three-agents:4}
        \begin{enumerate}
            \item In the partition caused by $(z, y)$ the bundles $[0, z]$ and $[z, y]$ are good for $u$, and $z < y$.\label{item:lem:x-y-three-agents:4a}
            \item In the partition caused by $(x, z)$ the bundles $[x, z]$ and $[z, m]$ are good for $u$.\label{item:lem:x-y-three-agents:4b}
        \end{enumerate}
    \end{enumerate}
\end{lemma}

\begin{proof}
    Assume that \ref{item:lem:x-y-three-agents:1} does not hold. It must then hold that $u([0, x]) < u([x, y])$. To show that \ref{item:lem:x-y-three-agents:3} holds, we must show that also $u([y, m]) < u([x, y])$. Note that \ref{item:lem:x-y-three-agents:2} follows directly from \ref{item:lem:x-y-three-agents:3} and \cref{lem:tree-monotone-better-than-neighbours}.
    
    If $u([0, c_1]) = u([0, x])$ it must hold that $u([x, y]) > u([c_1, c_2])$. Otherwise, both $[c_1, x]$ and $[c_2, y]$ only contain items with zero utility, and $([0, x], [x, y])$ must form a PMMS partition of $[0, y]$, causing \ref{item:lem:x-y-three-agents:1} to hold. Since $u([c_1, x]) = 0$ and $(c_1, c_2)$ causes a PMMS partition for $u$, we know that $(x, c_2)$ also causes a PMMS partition for $u$. Since $u([x, y]) > u([c_1, c_2]) = u([x, c_2])$ it follows from \cref{cor:conditions-p-mms-split} that $u([x, y]) > u([y, m])$, and \ref{item:lem:x-y-three-agents:3} holds when $u([0, c_1]) = u([0, x])$.

    Now assume that $u([0, c_1]) < u([0, x])$. By \cref{lem:conditions-p-mms-split} it must hold that $u([0, x]) \ge u([c_1, c_2])$. Additionally, it must hold that $u([c_2, y]) > 0$, as otherwise by \cref{cor:conditions-p-mms-split} we have $u([0, x]) > u([x, c_2]) = u([x, y])$ and \ref{item:lem:x-y-three-agents:1} holds. Thus, by \cref{lem:conditions-p-mms-split} it holds that $u([y, m]) \le u([c_1, c_2])$, and we get that $u([x, y]) > u([0, x]) \ge u([c_1, c_2]) \ge u([y, m])$. Thus, \ref{item:lem:x-y-three-agents:3} holds when \ref{item:lem:x-y-three-agents:1} does not.

    Continuing to assume that \ref{item:lem:x-y-three-agents:1} does not hold, we wish to show that given \ref{item:lem:x-y-three-agents:2} and \ref{item:lem:x-y-three-agents:3}, \ref{item:lem:x-y-three-agents:4} holds. We will show that if $u([0, x]) \ge u([y, m])$, then \ref{item:lem:x-y-three-agents:4a} holds, and otherwise \ref{item:lem:x-y-three-agents:4b} holds.
    
    Assume that $u([0, x]) \ge u([y, m])$. We wish to show that there exists $z$ with $x < z < y$ such that in the partition caused by $(z, y)$ the bundles $[0, z]$ and $[z, y]$ are good for $u$. Since \ref{item:lem:x-y-three-agents:1} does not hold, for any PMMS partition $([0, z], [z, y])$ of $[0, y]$ for $u$ it must be that $\min\{u([0, z]), u([z, y])\} > u([0, x]) \ge u([y, m])$. Moreover, by monotonicity, $z > x$ and since $u([z, y]) > 0$, $z < y$. Since $u([z, y]) > u([y, m])$ it follows from \cref{lem:tree-monotone-better-than-neighbours} and that $([0, z], [z, y])$ is a PMMS partition, that both $[0, z]$ and $[z, y]$ are good for $u$ in the partition caused by $(z, y)$. Consequently, \ref{item:lem:x-y-three-agents:4a} holds in this case.

    Now, assume that $u([0, x]) < u([y, m])$. By \ref{item:lem:x-y-three-agents:3} $u([x, y]) > u([y, m])$ and for any PMMS partition $([x, z], [z, m])$ of $[x, m]$, it holds that 
    $$ 
    \min\{u([x, z]), u([z, m])\} \ge u([y, m]). 
    $$    
    Moreover, $u([y, m]) > u([0, x]) \ge 0$. Thus, neither of the bundles are empty and it holds that $x < z < m$. By monotonicity, there must be at least one PMMS partition of $[x, m]$ with $z \le y$, as either $([x, y], [y, m])$ is a PMMS partition, or $\min\{u([x, z]), u([z, m])\} > u([z, m])$ for any PMMS partition. Let $z$ be such that $([x, z], [z, m])$ is a PMMS partition of $[x, m]$ and $z \le y$. Then
    $u([x, z]) \ge u([y, m]) > u([0, x])$ and both $[x, z]$ and $[z, m]$ must be good for $u$ in the partition caused by $(x, z)$. Since $x < z \le y$, \ref{item:lem:x-y-three-agents:4b} holds in this case.
\end{proof}

We now have all the tools needed to show PMMS existence in the 3-agent case on a path.



\theoremThreeAgentsPath*

\begin{table}[t]
    \centering
    \begin{tabular}{ccc}
        \toprule
        \textbf{} & \textbf{Cut Order} & \textbf{Covering Case} \\
        \midrule
        \multirow{5}{*}{$c_2^2 \le c_1^3$} & $(c_1^1, c_2^1, c_1^2, c_2^2, c_1^3, c_2^3)$ & \multirow{5}{*}{Case 1} \\
        & $(c_1^1, c_1^2, c_2^1, c_2^2, c_1^3, c_2^3)$ & \\
        & $(c_1^1, c_1^2, c_2^2, c_2^1, c_1^3, c_2^3)$ & \\
        & $(c_1^1, c_1^2, c_2^2, c_1^3, c_2^1, c_2^3)$ & \\
        & $(c_1^1, c_1^2, c_2^2, c_1^3, c_2^3, c_2^1)$ & \\  
        \midrule
        \multirow{4}{*}{\begin{tabular}{c}$c_2^2 > c_1^3$ \\ $c_2^1 < c_1^3$\end{tabular}} & $(c_1^1, c_2^1, c_1^2, c_1^3, c_2^3, c_2^2)$ & \multirow{2}{*}{Case 2}\\
        & $(c_1^1, c_1^2, c_2^1, c_1^3, c_2^3, c_2^2)$ & \\ \cmidrule{2-3}
        & $(c_1^1, c_2^1, c_1^2, c_1^3, c_2^2, c_2^3)$ & Case 3 \\ \cmidrule{2-3}
        & $(c_1^1, c_1^2, c_2^1, c_1^3, c_2^2, c_2^3)$ & Case 4 \\ 
        \midrule
        \multirow{6}{*}{\begin{tabular}{c}$c_2^2 > c_1^3$ \\ $c_2^1 \ge c_1^3$\end{tabular}} & $(c_1^1, c_1^2, c_1^3, c_2^1, c_2^2, c_2^3)$ & Case 4 \\ \cmidrule{2-3}
        & $(c_1^1, c_1^2, c_1^3, c_2^1, c_2^3, c_2^2)$ & Case 5 \\ \cmidrule{2-3}
        & $(c_1^1, c_1^2, c_1^3, c_2^2, c_2^1, c_2^3)$ & \multirow{2}{*}{Case 6} \\
        & $(c_1^1, c_1^2, c_1^3, c_2^2, c_2^3, c_2^1)$ & \\ \cmidrule{2-3}
        & $(c_1^1, c_1^2, c_1^3, c_2^3, c_2^1, c_2^2)$ & \multirow{2}{*}{Case 7} \\
        & $(c_1^1, c_1^2, c_1^3, c_2^3, c_2^2, c_2^1)$ & \\
        \bottomrule
    \end{tabular}
    \caption{Cases in the proof of \cref{thr:three-agents-path}.}
    \label{tab:proof-three-agents-path-cases}
\end{table}

\begin{proof}
    Fix some PMMS partition for each agent and let $(c_1^1, c_2^1)$, $(c_1^2, c_2^2)$ and $(c_1^3, c_2^3)$ be the tuples that cause these PMMS partitions for agent 1, 2 and 3, respectively. We assume w.l.o.g.\ that $c_1^1 \le c_1^2 \le c_1^3$, as otherwise the agents can be renamed in this manner. We now consider seven different cases, depending on the value of $c_2^1$, $c_2^2$ and $c_2^3$ in relation to each other. These seven cases cover all 15 possible orderings $(a, b, c, d, e, f)$ of the cut points in the PMMS partitions, when $c_1^1 \le c_1^2 \le c_1^3$, as can be seen from \cref{tab:proof-three-agents-path-cases}. The orderings are also illustrated in \cref{fig:cases-path} in the appendix, where the PMMS partition of agent 3 is on top (yellow), agent 2 in the middle (red) and agent 1 at the bottom (blue).

    \paragraph{Case 1: $\mathbf{c_2^2 \le c_1^3}$} We know that $[c_1^3, c_2^3] \subseteq [c_1^3, m] \subseteq [c_2^2, m]$ and $[c_1^2, c_2^2] \subseteq [0, c_1^3]$. By \cref{lem:sub-and-super-paths} it follows that the bundle $[c_2^2, m]$ is good for agent 3 in the partition caused by $(c_1^2, c_2^2)$. If $c_2^2 \le c_2^1$, we have that $[0, c_1^1] \subseteq [0, c_1^2]$ and $[c_1^2, c_2^2] \subseteq [c_1^1, c_2^1]$. By \cref{lem:sub-and-super-paths} it follows that the bundle $[0, c_1^2]$ is good for agent 1 in the partition $(c_1^2, c_2^2)$. Thus, a PMMS connected allocation exists by \ref{item:lem:3-agents-sufficient-conditions:3} in \cref{lem:3-agents-sufficient-conditions}. Else $c_2^2 > c_2^1$ and by \cref{lem:x-y-three-agents} at least one of bundles $[0, c_1^2]$ and $[c_1^2, c_2^2]$ is good for agent 1 in the partition $(c_1^2, c_2^2)$ and a PMMS connected allocation exists by \ref{item:lem:3-agents-sufficient-conditions:2} or \ref{item:lem:3-agents-sufficient-conditions:3} in \cref{lem:3-agents-sufficient-conditions}.

    \paragraph{Case 2: $\mathbf{c_2^1 < c_1^3 < c_2^2}$ and $\mathbf{c_2^3 \le c_2^2}$}
    We know that $[c_1^3, c_2^3] \subseteq [c_1^2, c_2^2]$. Thus, both $[0, c_1^2] \subseteq [0, c_1^3]$ and $[c_2^2, m] \subseteq [c_2^3, m]$. By \cref{lem:sub-and-super-paths} both $[0, c_1^3]$ and $[c_2^3, m]$ are good for agent 2 in the partition $(c_1^3, c_2^3)$. Since $(c_1^3, c_2^3)$ is a PMMS partition of agent 3, a PMMS connected allocation exists by \ref{item:lem:3-agents-sufficient-conditions:2} in \cref{lem:3-agents-sufficient-conditions}.

    \paragraph{Case 3: $\mathbf{c_2^1 < c_1^3 < c_2^2 < c_2^3}$ and $\mathbf{c_2^1 \le c_1^2}$}
    We know that $[c_1^2, m] \subseteq [c_2^1, m]$ and $[c_1^1, c_2^1] \subseteq [0, c_1^2]$. Thus by \cref{lem:sub-and-super-paths}, $[0, c_1^2]$ is good for agent 1 in the partition caused by $(c_1^2, c_2^2)$. Since $c_1^2 \le c_1^3$ and $c_2^2 < c_2^3$ it follows from \cref{lem:x-y-three-agents} that at least one of $[c_1^2, c_2^2]$ and $[c_2^2, m]$ is good for agent 3 in the partition caused by $(c_1^2, c_2^2)$. Thus, a PMMS connected allocation exists by \ref{item:lem:3-agents-sufficient-conditions:2} or \ref{item:lem:3-agents-sufficient-conditions:3} in \cref{lem:3-agents-sufficient-conditions}.

    \paragraph{Case 4: $\mathbf{c_2^1 < c_2^2 < c_2^3}$ and $\mathbf{c_2^2 \ge c_1^3}$} Since $c_1^1 \le c_1^2$ and $c_2^1 \le c_2^2$ we know by \cref{lem:x-y-three-agents} that at least one of $[0, c_1^2]$ and $[c_1^2, c_2^2]$ is good for agent 1 in the partition caused by $(c_1^2, c_2^2)$. Similarly, since $c_1^2 \le c_1^3$ and $c_2^2 \le c_2^3$ we know by \cref{lem:x-y-three-agents} that at least one of $[c_2^2, m]$ and $[c_1^2, c_2^2]$ is good for agent 3 in the partition caused by $(c_1^2, c_2^2)$. If the case is that $[0, c_1^2]$ is good for agent 1, then a PMMS allocation exists by \ref{item:lem:3-agents-sufficient-conditions:3} in \cref{lem:3-agents-sufficient-conditions}. Otherwise, by \ref{item:lem:x-y-three-agents:4} in \cref{lem:x-y-three-agents} one of \ref{item:lem:x-y-three-agents:4a} and \ref{item:lem:x-y-three-agents:4b} holds. If \ref{item:lem:x-y-three-agents:4a} holds there exists a $z$ with $c_1^2 < z < c_2^2$ such that both $[0, z]$ and $[z, c_2^2]$ are good for agent 1 in the partition caused by $(z, c_2^2)$. Since $[z, c_2^2] \subset [c_1^2, c_2^2]$ and $[0, c_1^2] \subset [0, z]$, by \cref{lem:sub-and-super-paths} both $[0, z]$ and $[c_2^2, m]$ are good for agent 2 in the partition caused by $(z, c_2^2)$. If \ref{item:lem:x-y-three-agents:4b} holds there exists a $z$ with $c_1^2 < z \le c_2^2$ both $[c_1^2, z]$ and $[z, m]$ are good for agent 1 in the partition caused by $(c_1^2, z)$. Since $[c_1^2, z] \subset [c_1^2, c_2^2]$ and $[c_2^2, m] \subset [z, m]$, by \cref{lem:sub-and-super-paths} both $[0, c_1^2]$ and $[z, m]$ are good for agent 2 in the partition caused by $(c_1^2, z)$. Since agents 1 and 2 prefer different pairs of bundles in the suggested partition in both cases, a PMMS connected allocation exist by \ref{item:lem:3-agents-sufficient-conditions:1} in \cref{lem:3-agents-sufficient-conditions}.
    
    \vspace{\baselineskip}
    \noindent We can now in the remaining cases assume that $c_1^3 \le c_2^1$ and $c_1^3 \le c_2^2$.
    
    \paragraph{Case 5: $\mathbf{c_2^1 \le c_2^3 \le c_2^2}$}
    We know that both $[c_2^2, m] \subseteq [c_2^3, m]$ and $[c_1^3, c_2^3] \subseteq [c_1^2, c_2^2]$. By \cref{lem:sub-and-super-paths} the bundle $[c_2^3, m]$ is good for agent 2 in the partition caused by $(c_1^3, c_2^3)$. Since $c_1^1 \le c_1^3$ and $c_2^1 \le c_2^3$ by \ref{item:lem:x-y-three-agents:1} and \ref{item:lem:x-y-three-agents:2} in \cref{lem:x-y-three-agents} we know that at least one of $[0, c_1^3]$ and $[c_1^3, c_2^3]$ is good for agent 1 in the partition caused by $(c_1^3, c_2^3)$. Thus, a PMMS connected allocation exists by \ref{item:lem:3-agents-sufficient-conditions:2} or \ref{item:lem:3-agents-sufficient-conditions:3} in \cref{lem:3-agents-sufficient-conditions}.

    \paragraph{Case 6: $\mathbf{c_2^3 \le c_2^1}$ and $\mathbf{c_2^3 \le c_2^2}$}
    We know that $[c_1^3, c_2^3] \subseteq [c_1^1, c_2^1]$ and $[c_1^3, c_2^3] \subseteq [c_1^2, c_2^2]$. Thus, it follows that $[0, c_1^1] \subseteq [0, c_1^3]$ and $[c_2^2, m] \subseteq [c_2^3, m]$. By \cref{lem:sub-and-super-paths} we therefore know that $[0, c_1^3]$ is good for agent 1 and $[c_2^3, m]$ for agent 2 in the partition caused by $(c_1^3, c_2^3)$. Thus, a PMMS connected allocation exists by \ref{item:lem:3-agents-sufficient-conditions:2} or \ref{item:lem:3-agents-sufficient-conditions:3} in \cref{lem:3-agents-sufficient-conditions}.

    \paragraph{Case 7: $\mathbf{c_2^2 \le c_2^1}$ and $\mathbf{c_2^2 \le c_2^3}$}
    We know that $[0, c_1^1] \subseteq [0, c_1^2]$ and $[c_1^2, c_2^2] \subseteq [c_1^1, c_2^1]$. By \cref{lem:sub-and-super-paths}, the bundle $[0, c_1^2]$ is good for agent 1 in the partition caused by $(c_1^2, c_2^2)$. Since $c_1^2 \le c_1^3$ and $c_2^2 \le c_2^3$, by \ref{item:lem:x-y-three-agents:1} and \ref{item:lem:x-y-three-agents:2} in \cref{lem:x-y-three-agents} we know that at least one of $[c_1^2, c_2^2]$ or $[c_2^2, m]$ is good for agent 3 in the partition caused by $(c_1^2, c_2^2)$. Thus, a PMMS connected allocation exists by \ref{item:lem:3-agents-sufficient-conditions:2} or \ref{item:lem:3-agents-sufficient-conditions:3} in \cref{lem:3-agents-sufficient-conditions}.
    \vspace{\baselineskip}
    
    \noindent    
    \Cref{lem:SMMS-additive-trees-polynomial} allows us to find a PMMS partition for each of the three agents in polynomial time. Given these PMMS partitions, the order of the cut points and consequently the applicable case can easily be determined. We can verify that each of the cases can be performed in polynomial time, as they consist of checking if simple properties hold and in certain cases moving a cut point along a path until some simple properties hold. Therefore, a PMMS allocation can be found in polynomial time.
\end{proof}
 
\section{Identical utilities}\label{sec:identical}

When agents have identical utilities, a PMMS and MMS connected allocation is guaranteed to exist for any graph. The proof is similar to that of Theorem~4.2 in \citet{Plaut:20}. 

\begin{theorem}\label{lem:identical-utility-functions}
For a connected graph $G$ and $n$ agents with identical utilities, a PMMS and MMS connected allocation exists.
\end{theorem}

\begin{proof}
    Let $A$ be a connected leximin allocation. Assume that $A$ is not PMMS. Let $i$ and $j$ be a pair of agents for which PMMS is not satisfied and $(A'_i, A'_j)$ a PMMS partition of $A_i \cup A_j$. Let $A'$ be the allocation obtained by exchanging $A_i$ and $A_j$ for $A'_i$ and $A'_j$ in $A$. Then $u(A'_k) = u(A_k)$ for every agent $k \in N \setminus \{i, j\}$, and $\min\{u(A'_i), u(A'_j)\} > \min\{u(A_i), u(A_j)\}$. Thus, $A$ is not leximin, a contradiction, and it follows that $A$ must be PMMS. Moreover, $A$ is MMS, as otherwise any MMS allocation, of which at least one exists for identical utilities, would be a leximin improvement of $A$.  
\end{proof}

\Cref{lem:identical-utility-functions} shows that a leximin allocation is both MMS and PMMS for identical utilities, no matter the utility function and graph. For many combinations of graph classes and utility functions, finding a leximin allocation is hard. However, with monotone utilities on a path, a leximin allocation can be found in polynomial time \cite{BiloCaFl22}. This yields the following corollary.

\begin{corollary}\label{cor:PMMS+MMS:path:monotone}
    For a path and $n$ agents with identical monotone utilities, a PMMS and MMS connected allocation can be found in polynomial time.
\end{corollary}

Note that if \cref{cor:PMMS+MMS:path:monotone} is relaxed to require only PMMS, there exists a moving-knife style algorithm that solves the problem (see \cref{sec:moving-knife}). Also, note that \cref{lem:identical-utility-functions} can be used to show the existence of a PMMS and MMS connected allocation when one of the agents is different from the rest and has additive utilities. If the graph is a tree, it is sufficient for PMMS that the utility function is monotone.

\begin{lemma}
    If all but one agent have identical utility function $u_1$, and the last agent has additive utility function $u_2$, a PMMS and MMS connected allocation exists.
\end{lemma}

\begin{proof}
    Let $A$ be a leximin allocation for $u_1$ and $A_i$ the bundle in $A$ with the highest utility for $u_2$. Since $u_2$ is additive, it holds that $u_2(A_i) \ge u_2(A_i \cup A_j)/2 \ge u_2(A_j)$ for every bundle $A_j$ adjacent to $A_i$. Further, $\PMMS(A_i \cup A_j) \le u_2(A_i \cup A_j)/2$. Hence, giving $A_i$ to the agent with utility function $u_2$ and the remaining bundles to the other agents is a PMMS connected allocation. Moreover, by \cref{lem:identical-utility-functions} the allocation is MMS if the agent with utility function $u_2$ receives a bundle with utility at least equal to her MMS. Since $u_2$ is additive, and $A_i$ is the bundle in the allocation with the highest utility for $u_2$, it follows that there cannot be any partition in which every bundle has a utility of surpassing $u_2(A_i)$. Consequently, the agent receives a bundle with utility at least equal to her MMS and the allocation is MMS.
\end{proof}

\begin{lemma}
    If all but one agent have identical utility function $u_1$, the last agent has monotone utility function $u_2$ and the graph is a tree, then a PMMS connected allocation exists.
\end{lemma}

\begin{proof}
    By \cref{lem:identical-utility-functions} a PMMS connected allocation exists when all agents have an identical utility function. Let $A$ be any such allocation for $n$ agents with utility function $u_1$. By \cref{lem:tree-monotone-at-least-one} there is at least one bundle $A_i \in A$ that is worth PMMS for $u_2$ when compared to any of its adjacent bundles in $A$. Thus, allocating $A_i$ to the agent with utility function $u_2$ and allocating the remaining $n - 1$ bundles to the remaining $n - 1$ agents in any way produces a PMMS connected allocation.
\end{proof}

\subsection{$\mathbf{3/4}$-PMMS on general graphs}\label{sec:PMMS:identical:anygraph}
Finding a leximin allocation is strongly NP-hard for identical utilities~\cite{Garey:90}.
Unless P$=$NP, this excludes the existence of a pseudo-polynomial time algorithm for finding a PMMS connected allocation via a leximin allocation even when agents have identical utilities. We show on the other hand that there exists a pseudo-polynomial time algorithm for a $(3/4)$-PMMS connected allocation based on a sequence of local improvements between pairs of neighboring agents.

\begin{restatable}{theorem}{thmPMMSpseudopoly}\label{thm:PMMS-pseudopolynomial}
For a connected graph $G$ and $n$ agents with identical additive utility function $u \colon V \rightarrow \mathbb{Z}_{+}$, a $(3/4)$-PMMS connected allocation can be found in pseudo-polynomial time.
\end{restatable}

\begin{proof}
For each connected subset $V'$ of vertices, we write 
$$
g(V')=\min \{u(X),u(Y)\}
$$
for a $(3/4)$-PMMS connected allocation $(X,Y)$ of $V'$ guaranteed to be polynomial-time computable in Theorem~\ref{thm:halfPMMS-polynomial:twoagents:identical} for the induced graph $G[V']$ and identical utility $u$. Consider Algorithm~\ref{alg:halfPMMS:n:identical}. 
\begin{algorithm}                      
\caption{Algorithm for constructing $3/4$-PMMS connected allocation for a connected graph $G$ and $n$ agents with identical utilities $u$}         
\label{alg:halfPMMS:n:identical}                          
\begin{algorithmic}[1]
\STATE Construct an arbitrary connected allocation $(A_1,A_2,\ldots,A_n)$. 
\WHILE{$u(A_i) < g(A_i \cup A_j)$ for some pair of agents $i$ and $j$ such that $A_i=\emptyset$ or $i$ and $j$ are adjacent under $A$}
\STATE Compute a $(3/4)$-PMMS connected allocation $(X,Y)$ for $G[A_i \cup A_j]$ and $u$. Assume $u(X) \leq u(Y)$.
\STATE Update $A_i=X$ and $A_j=Y$. 
\ENDWHILE
\end{algorithmic}
\end{algorithm}

We claim that $\sum_{i \in N}(u_i(A_i))^2$ decreases by at least $1$ after a polynomial number of steps during the execution of Algorithm~\ref{alg:halfPMMS:n:identical}. To see this, consider any pair of agents $i,j$. Suppose that $g(A_i \cup A_j)=u(A_i)+\alpha$ for some positive $\alpha>0$. This means that $u(A_j) > \frac{1}{2}u(A_i \cup A_j) \geq g(A_i \cup A_j)= u(A_i) + \alpha$.

Let $(X,Y)$ be a $(3/4)$-PMMS connected allocation for $G[A_i \cup A_j]$ and identical utility function $u$. 
Now consider the new allocation $A'$ where $A'_i=X$, $A'_j=Y$, and $A'_k=A_k$ for any agent $k \neq i,j$. We will show that $\sum_{i \in N}(u_i(A'_i))^2 < \sum_{i \in N}(u_i(A_i))^2$. 
\begin{align*}
&\sum_{i \in N}(u_i(A'_i))^2 - \sum_{i \in N}(u_i(A_i))^2\\
&=u(A'_i)^2+u(A'_j)^2 -  u(A_i)^2 - u(A_j)^2 \\
&=(u(A_i) + \alpha )^2+(u(A_j) - \alpha))^2 -  u(A_i)^2 - u(A_j)^2 \\
&= 2\alpha (u(A_i) - u(A_j) + \alpha) \leq 1. 
\end{align*}
Thus, the number of iterations of the {\bf while}-loop is at most $n u_i(V)^2 \leq n(|V| u_{max})^2$ where $u_{max} = \max_{v \in V}u(v)$. Further, by Theorem~\ref{thm:halfPMMS-polynomial:twoagents:identical}, each step of the {\bf while}-loop can be implemented in polynomial time. This establishes the claim. 
\end{proof}


\subsection{PMMS on trees}\label{sec:PMMS:trees}


Forgoing leximin, we can construct an algorithm that guarantees both PMMS and MMS for identical utilities. For this purpose, we make use of the following strengthening of MMS introduced by \citet{BiloCaFl22}.
If agents have identical utilities, an allocation $A$ is \emph{strongly maximin share} (SMMS) if it is MMS and minimizes the number of agents $i$ with $u(A_i) = \MMS(V)$ among all MMS allocations. An agent $i$ with $u(A_i) = \MMS(V)$ is called a \emph{loser}.
\Cref{alg:PMMS+SMMS} finds a PMMS and MMS connected allocation by finding an SMMS allocation, fixing the bundles of the losers and repeating the process for the remaining agents. Since SMMS is a stronger requirement than MMS, \cref{alg:PMMS+SMMS} is not polynomial in the general case, unless P$=$NP. However, for trees and additive utilities, we can show that the algorithm is polynomial. Note that the output of the algorithm may not be leximin since PMMS and SMMS may not imply leximin.

\begin{lemma}\label{lem:PMMS:SMMS:relation}
    Even for a path and three agents with identical additive utilities, PMMS and SMMS do not imply leximin.
\end{lemma}

\begin{proof}
    Consider an instance with three agents and five item on a path. Each agent has utility $1$, $3$, $1$, $1$ and $1$ for each of the items, starting from the leftmost. Consider a connected allocation $A$ that allocates the two leftmost items to agent $1$, the middle item to agent $2$ and the two rightmost items to agent $3$. Also consider a connected allocation $A'$ that allocates the leftmost item to agent $1$, the next item to agent $2$ and the three rightmost items to agent $3$.
    Both allocations are PMMS and SMMS. However, only $A'$ is leximin, as $A'$ is a leximin improvement of $A$.
\end{proof}

\begin{algorithm}                      
\caption{Algorithm for constructing a PMMS and SMMS connected allocation for a connected graph $G = (V, E)$ and $n$ agents with identical utilities $u$}
\label{alg:PMMS+SMMS}                          
\begin{algorithmic}[1]
\STATE Find a SMMS connected allocation $A$\label{line:initial-allocation} and create graph $G'$ from $G$ by removing the items allocated to losers in $A$.\label{line:initial-assignment}
\FOR{each connected component $C = (V_C, E_C)$ in $G'$}
\STATE Run \cref{alg:PMMS+SMMS} for $C$ and agents $N_C = \{\, i \mid A_i \subseteq C, i \in [n] \,\}$, and let $A_C$ be the resulting allocation.\label{line:recursive-call}
\STATE For each $i \in N_C$, swap $i$'s bundle in $A$ for $i$'s bundle in $A_C$.\label{line:update-A}
\ENDFOR 
\RETURN A\label{line:return-PMMS-SMMS}
\end{algorithmic}
\end{algorithm}

\begin{lemma}\label{lem:PMMS+SMMS-exists}
    For a connected graph $G$ and $n$ with identical utility function $u$, \cref{alg:PMMS+SMMS} finds a PMMS and SMMS connected allocation. It runs in polynomial time if connected SMMS allocations can be found in polynomial time and the utility $u(X)$ can be computed in polynomial time for each $X \subseteq V$.
\end{lemma}

\begin{proof}
    We prove PMMS and SMMS by induction on recursive calls to \cref{alg:PMMS+SMMS}. When every agent is a loser in $A$---our inductive base case---$A$ must be both SMMS and leximin, as any improvement would result in fewer losers. By \cref{lem:identical-utility-functions}, $A$ is thus PMMS.
    
    For our inductive step, we claim that if each $A_C$ found is a PMMS and SMMS connected allocation, then $A$ is a PMMS and SMMS connected allocation when returned on Line~\ref{line:return-PMMS-SMMS}. 
    Since an SMMS connected allocation is found on Line~\ref{line:initial-allocation}, every non-loser agent belongs to exactly one set $N_C$ and $A$ remains a connected allocation. Also, $\mu^{|N_C|}(V_C) > \mu^n(V)$ must hold, as $u(A_i) > \mu^n(V) = \MMS(V)$ for every $i \in N_C$. Thus, every agent $i$ that was not a loser on Line~\ref{line:initial-assignment} receives a bundle with $u(A_i) > \MMS(V)$ and $A$ remains SMMS.

    Since $A_C$ is PMMS, the only way PMMS can not hold in $A$ is between a loser $i$ and an agent $j \in N_C$ for some connected component $C$. Assume this is the case when $A$ is returned and let $(B_i, B_j)$ be a PMMS partition of $A_i \cup A_j$. Replacing $A_i$ and $A_j$ by $B_i$ and $B_j$ would reduce the number of losers in $A$ or increase the utility of the worst off agent, as $\min\{u(A_i), u(A_j)\} < \min\{u(B_i), u(B_j)\}$. This is a contradiction, as then $A$ is not SMMS. Hence, $A$ is PMMS.

    Since any SMMS allocation has at least one loser and every non-loser belongs to exactly one set $N_C$, \cref{alg:PMMS+SMMS} is ran at most $n$ times. Under the stated assumptions, we can easily verify that every operation in the algorithm can be performed in polynomial time. Thus, the algorithm must run in polynomial time.
\end{proof}


We now show that an SMMS connected allocation can be found for identical additive utilities on trees in polynomial time.

\begin{theorem}\label{lem:SMMS-additive-trees-polynomial}
    For a tree $G$ and $n$ agents with identical additive utilities, a SMMS connected allocation can be found in polynomial time.
\end{theorem}

Finding an SMMS allocation is trivial if $\MMS(V) = 0$; Simply give as many agents as possible items with non-zero utility. When $\MMS(V) > 0$, we solve a related problem:  First, root the tree in some arbitrary vertex $r$. For any vertex $v_i$ and pair of integers $0 \le j,\ell \le n$, we are interested in finding, if it exists, a partition $P_{i,j,\ell} = (B_1, \dots, B_{j + 1})$ of the subtree rooted in $v_i$ into connected bundles that subject to the following conditions maximizes $u(B_{j + 1})$:
\begin{enumerate}[label=(\roman*)]
    \item $B_{j + 1} = \emptyset$ or $v_i \in B_{j + 1}$ \label{item:SMMS:tree:T:1}
    \item $u(B_t) \ge \MMS(V)$ for $1 \le t \le j$ \label{item:SMMS:tree:T:2}
    \item $|\{\, B_t \mid u(B_t) = \MMS(V), 1 \le t \le j \,\}| \le \ell$ \label{item:SMMS:tree:T:3}
\end{enumerate}

\noindent
In other words, $P_{i,j,\ell}$ partitions the subtree into $j$ bundles with utility at least $\MMS(V)$ of which at most $\ell$ have a utility of exactly $\MMS(V)$. Finally, the remaining bundle, $B_{j + 1}$, is either empty or contains $v_i$. This property is important, as $B_{j + 1}$, the only bundle not guaranteed to have a utility of at least $\MMS(V)$, can always be combined with a bundle containing $v_i$'s parent. Together with maximizing $u(B_{j + 1})$, this allows for solutions to be found by dynamic programming. Note that since $G$ is a tree, $\MMS(V)$ can be found in polynomial time using an algorithm of \citet{Perl:81}. We now show that solving the problem for $i = r$, $j = n$ and $0 \le \ell \le n$ yields an SMMS connected allocation.

\begin{lemma}\label{lem:SMMS:tree:construct}
    For a tree $G = (V, E)$ rooted in some vertex $r \in V$ and $n$ agents with identical utilties, $P_{r, n, \ell}$ exists for some $0 \le \ell \le n$ and a SMMS connected allocation of $G$ can be obtained from $P_{r, n, \ell^*} = (B_1, \dots, B_n, B_{n + 1})$ by removing $B_{n + 1}$ and redistributing the items in $B_{n + 1}$ to the other bundles in any valid way,  where $\ell^*$ is the smallest $0 \le \ell \le n$ for which $P_{r, n, \ell}$ exists.
\end{lemma}


\begin{proof}
    Let $A = (A_1, \dots, A_n)$ be an SMMS connected allocation for $G$ and $\ell_A = |\{\, A_i \mid u(A_i) = \MMS(V), i \in [n] \,\}|$. We claim that $A' = (A_1, \dots, A_n, \emptyset)$ satisfies \ref{item:SMMS:tree:T:1}--\ref{item:SMMS:tree:T:3} for the triple $r, n, \ell_A$. Indeed $B_{j + 1} = \emptyset$ and \ref{item:SMMS:tree:T:1} holds. Since $A$ is SMMS, \ref{item:SMMS:tree:T:2} holds. By the definition of $\ell_A$, \ref{item:SMMS:tree:T:3} also holds. Thus, $P_{r, n, \ell}$ exists for at least one $\ell$.

    Let $B = (B'_1, \dots, B'_n)$ be the allocation obtained from $P_{r, n, \ell^*}$ by removing $B_{n + 1}$ and redistributing the items. By \ref{item:SMMS:tree:T:2}, every bundle $B'_i \in B$ has $u(B'_i) \ge \MMS(V)$. Since $\ell^*$ is minimized, it holds that $\ell^* \le \ell_A$ and by \ref{item:SMMS:tree:T:3} the number of bundles with utility exactly $\MMS(V)$ is no greater in $B$ than in $A$. In fact, it must hold that $\ell^* = \ell_A$ as otherwise $A$ is not SMMS. Thus, $B$ must be SMMS.
\end{proof}

To solve the problem for the selected root $r$, $j = n$ and $0 \le \ell \le n$, we will rely on dynamic programming to combine solutions for the children of $r$. This process is repeated until leaf vertices are reached.


\begin{restatable}{lemma}{lemSMMSTreeInternalNode}\label{lem:SMMS:tree:internal:node}
 For a vertex $v_i$, $P_{i, j, \ell}$ can be computed in polynomial time for any $0 \le j,\ell \le n$ if $P_{h, j_h, \ell_h}$ is known for every child vertex $v_h$ of $v_i$ and pair $0 \le j_h, \ell_h \le n$.
\end{restatable}

In order to prove \cref{lem:SMMS:tree:internal:node}, we rely on several intermediate results. First, note that when $j = 0$, a solution always exists and is trivial to find.

\begin{lemma}\label{lem:SMMS:tree:j:0}
    For a vertex $v_i$, $j = 0$ and any $0 \le \ell \le n$, there is a unique solution to $P_{i, j, \ell}$, the solution in which $B_1$ contains the entire subtree rooted in $v_i$.
\end{lemma}

\begin{proof}
    Since $j = 0$, $P_{i, j, \ell}$ contains a single bundle. Therefore, there is only a single possible partition, the one in which all items appear in the same bundle. This partition satisfies \ref{item:SMMS:tree:T:1}--\ref{item:SMMS:tree:T:3}. As no other partitions are possible, this partition is a solution to $P_{i, j, \ell}$.
\end{proof}

Next, we show that whenever a solution exists for some subproblem $P_{i, j, \ell}$ one of two conditions hold. Either condition is also a sufficient condition for a solution to $P_{i, j, \ell}$ to exist.

\begin{lemma}\label{lem:SMMS:internal:cases}
    For an internal vertex $v_i$, $1 \le j \le n$ and $0 \le \ell \le n$, a solution exists for $P_{i, j, \ell}$ if and only if at least one of two following cases holds:

    \begin{enumerate}[label=(\alph*)]
        \item There exists a combination of pairs $j_h, \ell_h$, one for each child $v_h$ of $v_i$, with $\sum_{h} j_h = j$ and $\sum_{h} \ell_h = \ell$ such that a solution to $P_{h, j_h, \ell_h}$ exists for every child $v_h$.\label{item:SMMS:tree:internal:structure:1}
        \item There exists (1) a solution to $P_{i, j - 1, \ell}$ where the $j$-th bundle has utility exceeding $\MMS(V)$ or (2) a solution to $P_{i, j - 1, \ell - 1}$ where the $j$-th bundle has a utility of at least $\MMS(V)$.\label{item:SMMS:tree:internal:structure:2}
    \end{enumerate}
\end{lemma}

\begin{proof}
    We prove each direction.
    
    \paragraph{$\implies$} Assume that \ref{item:SMMS:tree:internal:structure:1} does not hold, as otherwise the claim holds.
    Then, we claim that in any solution for $P_{i, j, \ell}$ the $(j + 1)$-th bundle is empty. Assume that this is not the case and let $P = (B_1, \dots, B_{j + 1})$ be some solution for $P_{i, j, \ell}$. Since $G$ is a tree, every bundle except $B_{j + 1}$ contains only items from a subtree rooted in a single child $v_h$ of $v_i$. Thus, for every child $v_h$ there is a number of bundles $0 \le j_h \le j$ in $P$ that contains only items from the subtree rooted in $v_h$. Let $\ell_h$ denote the number of these bundles that have utility exactly $\MMS(V)$. Note that by \ref{item:SMMS:tree:T:2} and \ref{item:SMMS:tree:T:3}, it must hold that $\sum_h \ell_h = \ell$ and $\sum_{h} j_h = j$. Moreover, a solution to $P_{h, j_h, \ell_h}$ must exist, as the partition created by taking the $j_h$ bundles in $P$ that belong only to the subtree rooted in $v_h$ and redistributing any items in $B_{j + 1}$ from the subtree in any valid way satisfies \ref{item:SMMS:tree:T:1}--\ref{item:SMMS:tree:T:3} for $j_h$ and $\ell_h$. This is a contradiction, as then \ref{item:SMMS:tree:internal:structure:1} must hold.

    Let $P = (B_1, \dots, B_j, \emptyset)$ be some solution for $P_{i, j, \ell}$. Moreover, let $B_j$ be the bundle that contains $v_i$ and $P' = (B_1, \dots, B_j)$.  
    If $u(B_j) = \MMS(V)$, then $P'$ satisfies \ref{item:SMMS:tree:T:1}--\ref{item:SMMS:tree:T:3} for $P_{i, j - 1, \ell - 1}$ and a solution must exist for $P_{i, j - 1, \ell - 1}$. Similarly, if $u(B_j) > \MMS(V)$, then $P'$ satisfies \ref{item:SMMS:tree:T:1}--\ref{item:SMMS:tree:T:3} for $P_{i, j - 1, \ell}$ and a solution must exist for $P_{i, j - 1, \ell}$. Since $u(B_j) \ge \MMS(V)$ by \ref{item:SMMS:tree:T:3}, at least one of these two cases must hold. As any solution to $P_{i, j - 1, \ell}$ or $P_{i, j - 1, \ell - 1}$ maximizes the utility of the $j$-th bundle, at least one of the two conditions in \ref{item:SMMS:tree:internal:structure:2} holds when \ref{item:SMMS:tree:internal:structure:1} does not.

    \paragraph{$\impliedby$} Assume that \ref{item:SMMS:tree:internal:structure:1} holds. For every child $v_h$ of $v_i$ fix some solution $P_{h, j'_h, \ell'_h} = (B(h, j'_h, \ell'_h)_1, \dots, B(h, j'_h, \ell'_h)_{(j'_h + 1)})$, if one exists, for every pair $0 \le j'_h, \ell'_h \le n$. We claim that the partition consisting of the bundles $B(h, j_h, \ell_h)_1, \dots, B(h, j_h, \ell_h)_{j_h}$ for every child $v_h$ of $v_i$ along with the bundle 
    \[B_{j + 1} = \{v_i\} \cup \left(\bigcup_h B(h, j_h, \ell_h)_{(j_h + 1)}\right)\]
    satisfies \ref{item:SMMS:tree:T:1}--\ref{item:SMMS:tree:T:3} for $P_{i,j,\ell}$. Indeed, there are $1 + \sum_h j_h = j + 1$ bundles of which excluding $B_{j + 1}$ at most $\sum_{h} \ell_h = \ell$ have a utility of exactly $\MMS(V)$, and \ref{item:SMMS:tree:T:3} holds. Moreover, every bundle except $B_{j + 1}$ must have a utility of at least $\MMS(V)$, and \ref{item:SMMS:tree:T:2} holds. Finally, since $v_i \in B_{j + 1}$, \ref{item:SMMS:tree:T:1} holds. Thus, a solution exists for $P_{i, j, \ell}$.

    Assume now that \ref{item:SMMS:tree:internal:structure:2} holds. If (1) holds, then it can easily be verified that adding an empty bundle to this solution for $P_{i, j - 1, \ell}$ satisfies \ref{item:SMMS:tree:T:1}--\ref{item:SMMS:tree:T:3} for $P_{i, j, \ell}$. Similarly if (2) holds, \ref{item:SMMS:tree:T:1}--\ref{item:SMMS:tree:T:3} is satisfied for $P_{i, j, \ell}$ by adding an empty bundle to this solution for $P_{i, j - 1, \ell - 1}$. Thus, in either case $P_{i, j, \ell}$ exists.
\end{proof}

Next we show that if a solution exists, then one can be constructed in polynomial time in one of two ways, depending on if \ref{item:SMMS:tree:internal:structure:1} from \cref{lem:SMMS:internal:cases} holds.

\begin{lemma}\label{lem:SMMS:tree:b:construction}
    Let $v_i$ be an internal vertex and $1 \le j \le n$, $0 \le \ell \le n$ such that a solution exists for $P_{i, j, \ell}$ and \ref{item:SMMS:tree:internal:structure:1} from \cref{lem:SMMS:internal:cases} does not hold. Then a solution for $P_{i, j, \ell}$ can be constructed in polynomial time if a solution for $P_{i, j - 1, \ell}$ is known and if $\ell > 0$, a solution for $P_{i, j - 1, \ell - 1}$ is known.
\end{lemma}

\begin{proof}
    From the proof of \cref{lem:SMMS:internal:cases} we know that in any solution $P = (B_1, \dots, B_j, B_{j + 1})$ for $P_{i, j, \ell}$, that $B_{j + 1} = \emptyset$. Thus, $u(B_{j + 1}) = u(\emptyset) = 0$ and any partition satisfying \ref{item:SMMS:tree:T:1}--\ref{item:SMMS:tree:T:3} for $P_{i, j, \ell}$ is a solution for $P_{i, j, \ell}$.
    
    The second part of the proof of \cref{lem:SMMS:internal:cases} provides a way to construct a partition satisfying \ref{item:SMMS:tree:T:1}--\ref{item:SMMS:tree:T:3} for $P_{i, j, \ell}$ that can easily be verified to be polynomial when solutions are known for $P_{i, j - 1, \ell}$ and, if $\ell > 0$, for $P_{i, j - 1, \ell - 1}$.
\end{proof}

\begin{lemma}\label{lem:SMMS:tree:a:construction}
    Let $v_i$ be an internal vertex and $1 \le j \le n$, $0 \le \ell \le n$ such that a solution for $P_{i, j, \ell}$ exists and \ref{item:SMMS:tree:internal:structure:1} from \cref{lem:SMMS:internal:cases} holds. Then a solution for $P_{i, j, \ell}$ can be constructed in polynomial time if solutions for $P_{h, j_h, \ell_h}$ are known for every child $v_h$ of $v_i$ and pair $0 \le j_h,\ell_h \le n$, where a solution exists.
\end{lemma}

\begin{proof}
    For every child $v_h$ of $v_i$ fix for every pair $0 \le j'_h, \ell'_h \le n$ some solution $P_{h, j'_h, \ell'_h} = (B(h, j'_h, \ell'_h)_1, \dots, B(h, j'_h, \ell'_h)_{(j'_h + 1)})$, if one exists. To construct a solution, find some combination of pairs $j^*_h, \ell^*_h$, one for each child $v_h$ of $v_i$, that maximizes
    \begin{equation}\label{eq:maximize}
        \sum_{h} u(B(h, j^*_h, \ell^*_h)_{(j^*_h + 1)})
    \end{equation}%
    across all such combinations of pairs with $\sum_{h} j^*_h = j$ and $\sum_{h} \ell^*_h = \ell$, and for every child $v_h$ a solution for $P_{h, j^*_h, \ell^*_h}$ exists. Then, we claim that the partition consisting of bundles $B(h, j^*_h, \ell^*_h)_1, \dots, B(h, j^*_h, \ell^*_h)_{j^*_h}$ for each child $v_h$ along with the bundle
    \[B_{j + 1} = \{v_i\} \cup \left(\bigcup_h B(h, j^*_h, \ell^*_h)_{(j^*_h + 1)}\right)\]
    is a solution to $P_{i, j, \ell}$.
    
    It was shown in the second part of the proof of \cref{lem:SMMS:internal:cases} that the partition satisfies \ref{item:SMMS:tree:T:1}--\ref{item:SMMS:tree:T:3}. The only way that the partition is not a solution to $P_{i, j, \ell}$ is if $u(B_{j + 1})$ is not maximum under \ref{item:SMMS:tree:T:1}--\ref{item:SMMS:tree:T:3}. Assume that this is the case and let $P' = (B'_1, \dots, B'_{j + 1})$ be some solution for $P_{i, j, \ell}$. Since $u(B'_{j + 1}) > u(B'_{j + 1}) \ge 0$ it holds that $B'_{j + 1} \neq \emptyset$. Using the same decomposition argument for $P'$ as was used in the first part of the proof of \cref{lem:SMMS:internal:cases}, it follows that there exists a combination of pairs $j'_h, \ell'_h$, one for each child $v_h$ of $v_i$, with $\sum_{h} j'_h = j$, $\sum_{h} \ell'_h = \ell$, for every child $v_h$ of $v_i$ a solution for $P_{h, j'_h, \ell'_h}$ exists and
    \[u(B'_{j + 1}) = u(\{v_i\}) + \sum_{h} u\left(B(h, j'_h, \ell'_h)\right)\;.\]
    As $B'_{j + 1}$ is constructed in the same way as $B_{j + 1}$ in our partition was constructed, only with the pairs $j'_h, \ell'_h$. This is a contradiction, as the chosen combination of pairs $j^*_h, \ell^*_h$ cannot have maximized \cref{eq:maximize}. Consequently, the proposed partition is a solution for $P_{i,j,\ell}$.

    Determining an optimal combination of $j^*_h$ and $\ell^*_h$ may take exponential time if all combinations of pairs are tried directly. Instead, we show that an optimal combination can be found in polynomial time using dynamic programming.

    For some $0 \le j' \le \ell$, $0 \le \ell' \le \ell$ and subset $C$ of children of $v_i$, let $U(C, j', \ell')$ denote the maximum utility of
    \begin{equation}\label{eq:maximize3}
        \sum_{v_h \in C} u(B(h, j_h, \ell_h)_{(j_h + 1)})
    \end{equation}
    over all combinations of pairs $j_h, \ell_h$, one for each child in $C$, such that $\sum_{v_h \in C} j_h = j'$, $\sum_{v_h \in C} \ell_h = \ell'$ and a solution to $P_{h, j_h, \ell_h}$ exists for every $v_h \in C$. If no such combination exists, then let $U(C, j', \ell') = -\infty$.
    
    Select some child $v_h$ of $v_i$. Then, an optimal choice of $j^*_h,\ell^*_h$ can be made by selecting the pair $0 \le j_h \le j$ and $0 \le \ell_h \le \ell$, where a solution to $P_{h, j_h, \ell_h}$ exists and that maximizes
    \begin{equation}\label{eq:maximize2}
        u(B(h, j_h, \ell_h)_{(j_h + 1)}) + U(C, j - j_h, \ell - \ell_h)
    \end{equation}
    over all such pairs, where $C$ is the set of all children of $v_i$ except $v_h$. This follows directly from the definition of $U$, as expanding $U$ in \cref{eq:maximize2} yields \cref{eq:maximize}.

    For every $0 \le j' \le j$ and $0 \le \ell' \le \ell$, we can find a solution to $U(C, j', \ell')$ in a similar way. Instead of solving $U(C, j', \ell')$, we can solve our original problem (\cref{eq:maximize}) for the subset $C$ of children, $j'$ and $\ell'$, since \cref{eq:maximize3} is in this case equivalent to \cref{eq:maximize}. In other words, we can solve our original problem at most $(n + 1)^2$ times with one less child. We can again select some new child $v_g \in C$ and solve the problem as for $v_h$. Notice that solving $U(C, j', \ell')$ now depends only on solutions to $U(C \setminus \{v_g\}, j'', \ell'')$, $0 \le j'' \le j'$ and $0 \le \ell'' \le \ell'$ for every $j'$ and $\ell'$. Thus, the problems share the same dependencies on $(n + 1)^2$ new smaller versions of our problem. If this process is repeated until the subset of children contains only a single child, a trivial case, a total of $\le r(n + 1)^2$ subproblems need to be solved, where $r$ is the number of children of $v_i$. Each subproblem can be solved in polynomial time when solutions are known for the subproblems it depends on, as there are at most $(n + 1)^2$ choices for the pair $j_h$ and $\ell_h$ in \cref{eq:maximize2}. Thus, if follows that an optimal combination for $j^*_h,\ell^*_h$ can be found in polynomial time.
\end{proof}

We are now ready to prove \cref{lem:SMMS:tree:internal:node}.

\lemSMMSTreeInternalNode*

\begin{proof}
    We consider two cases, depending on if $v_i$ is a leaf or internal vertex in the tree. Recall that $\MMS(V)$ can be found in polynomial time using a result of \citet{Perl:81}. For this reason, we assume for the remainder of the proof that $\MMS(V)$ is known. Additionally, recall that the problem only considers cases in which $\MMS(V) > 0$. Finally, we assume that $j \neq 0$, as then \cref{lem:SMMS:tree:j:0} provides a way to find a solution in polynomial time.

    \paragraph{$v_i$ is a leaf vertex}
    The subtree rooted in $v_i$ contains only $v_i$. Since $\MMS(V) > 0$, there are exactly two possible partitions that can satisfy \ref{item:SMMS:tree:T:2}, namely $(\{v_i\})$ and $(\{v_i\}, \emptyset)$. As there is a constant number of feasible partitions, we can for any pair $j, \ell$ easily check if \ref{item:SMMS:tree:T:1}--\ref{item:SMMS:tree:T:3} are satisfied for all feasible partitions in polynomial time. Consequently, a solution to $P_{i, j, \ell}$ can in polynomial time be computed or determined to not exist.
    
    \paragraph{$v_i$ is an internal vertex} \cref{lem:SMMS:internal:cases} provides a way to check if a solution exists for $P_{i, j, \ell}$. We can first check if \ref{item:SMMS:tree:internal:structure:1} holds. Notice that this can be done in polynomial time using the same dynamic programming approach as for finding optimal $j^*_h$ and $\ell^*_h$ in the proof of \cref{lem:SMMS:tree:a:construction}. If \ref{item:SMMS:tree:internal:structure:1} holds, then a solution can be found in polynomial time by \cref{lem:SMMS:tree:a:construction}.
    
    If \ref{item:SMMS:tree:internal:structure:1} does not hold, we can check \ref{item:SMMS:tree:internal:structure:2}. If \ref{item:SMMS:tree:internal:structure:2} does not hold, then by \cref{lem:SMMS:internal:cases} no solution exists for $P_{i, j, \ell}$.
    Checking \ref{item:SMMS:tree:internal:structure:2} and, if it holds, constructing a solution using \cref{lem:SMMS:tree:a:construction}, requires having found a solution for $P_{i, j - 1, \ell}$ and/or $P_{i, j - 1, \ell - 1}$ (only if $\ell \neq 0$) or determined that a solution does not exist either. As both \ref{item:SMMS:tree:internal:structure:2} and \cref{lem:SMMS:tree:a:construction} require the utility of the $j$-th bundle to be at least $\MMS(V) > 0$, the $j$-th bundle in any solution to $P_{i, j - 1, \ell}$ or $P_{i, j - 1, \ell - 1}$ cannot be empty if \ref{item:SMMS:tree:internal:structure:2} holds. Using the same logic as in the proof of \cref{lem:SMMS:internal:cases}, \ref{item:SMMS:tree:internal:structure:1} must hold for $P_{i, j - 1, \ell}$ or $P_{i, j - 1, \ell - 1}$ for this to be the case. Consequently, we can determine if \ref{item:SMMS:tree:internal:structure:2} holds and in that case find a solution by only considering \ref{item:SMMS:tree:internal:structure:1} for both $P_{i, j - 1, \ell}$ and $P_{i, j - 1, \ell - 1}$. As already seen, this can be done in polynomial time and it follows that we can find a solution to $P_{i, j, \ell}$ or determine that one does not exist in polynomial time.
\end{proof}

\begin{proof}[Proof of \cref{lem:SMMS-additive-trees-polynomial}]
    By \cref{lem:SMMS:tree:internal:node} we can compute $P_{i, j, \ell}$ in polynomial time for any vertex $v_i$ and pair $j, \ell$ if we have solved the problem for every child of $v_i$.
    Thus, $P_{i, j, \ell}$ can be computed through a post-order traversal of the tree. Since there is a polynomial number, $|V|(n + 1)^2$, of solutions to compute, we can determine all in polynomial time. Given $P_{r, n, \ell}$ for $0 \le \ell \le n$, an SMMS connected allocation can be constructed in polynomial time using \cref{lem:SMMS:tree:construct}.
\end{proof}

\begin{corollary}\label{cor:PMMS+SMMS+trees}
    For a tree $G$ and $n$ agents with identical additive utilities, a PMMS and SMMS connected allocation can be found in polynomial time.
\end{corollary}

With an approach similar to that of \citet{Truszczynski:20} for MMS in unicyclic graphs---graphs that contain at most a single cycle---\cref{lem:SMMS-additive-trees-polynomial,cor:PMMS+SMMS+trees} can be extended to unicyclic graphs (see \cref{app:identical}).

\begin{lemma}
    For a connected unicyclic graph $G$ and $n$ agents with identical additive utilities, a PMMS and SMMS connected allocation can be found in polynomial time.
\end{lemma}

\begin{proof}
    \Cref{lem:PMMS+SMMS-exists} guarantees that such an allocation exists and can be found in polynomial time if an SMMS allocation can be found in polynomial time. We will show that an SMMS allocation can be found for any unicyclic graph by applying the algorithm from \cref{lem:SMMS-additive-trees-polynomial} at most $m$ times. This follows a similar approach to that of \citet{Truszczynski:20} for determining the MMS of an agent in unicyclic graphs.
    
    Let $A$ be an SMMS allocation for a unicyclic graph $G$. Then either one of the edges in the cycle connects two bundles in $A$ or all the items on the cycle belong to the same bundle in $A$. In either case, there is at least one edge $e$ on the cycle that can be removed without making $A$ infeasible. Let $G_e$ by the graph obtained by removing $e$ from $G$. Then, $A$ is an SMMS allocation of $G_e$ and any SMMS allocation of $G_e$ is an SMMS allocation of $G$. While we do not know the edge $e$ apriori, one can simply remove each edge in the cycle in turn, and find an SMMS allocation for each of the resulting graphs. One of these allocations is SMMS for $G$. 
    This allocation can be found by comparing the allocations and selecting the one in which the worst-off agent receives the greatest utility and subject to this the one that minimizes the number of "losers". Since removing an edge from the cycle produces a tree, we can determine each of these SMMS allocations in polynomial time by \cref{lem:SMMS-additive-trees-polynomial}. Further, at most $m$ calls to \cref{lem:SMMS-additive-trees-polynomial} is required, as the cycle contains at most $m$ items. Thus, an SMMS allocation can be found in polynomial time.
\end{proof}

Unfortunately, it is unlikely that polynomial time algorithms for SMMS can be found in most cases with more complex graph classes or utility functions. 
Note that a utility function $u$ is \emph{subadditive} if $u(X) + u(Y) \ge u(X \cup Y)$ holds for every $X,Y\subseteq V$.

\begin{lemma}\label{lem:subadditive:binary:star:hard}
Even for a star and $n$ agents who have identical subadditive utility functions with binary marginal gains, finding an SMMS connected allocation may require an exponential number of utility queries on average.
\end{lemma}

\begin{proof}
    Let there be $n$ agents and consider the star graph $S_{2n - 2}$. Further, let $S$ be an arbitrary subset of $n - 1$ leaves and $v$ the central vertex of the star. We define the utility function $u$ as follows:
    \[
    u(B) =
    \begin{cases}
        2 & |B| > n \text{ or } B = S \cup \{v\} \\
        1 & |B| \le n \text{ and } S \notin B \\
        0 & B = \emptyset.
    \end{cases}
    \]
    One can easily verify that the utility function is subadditive and that the marginal gains are binary.

    Since the graph is a star, all but one agent receives a bundle consisting of at most a single leaf. Thus, at most one agent may receive a bundle with utility greater than $1$. In fact, there is a single SMMS allocation. In this allocation one agent receives the bundle $B = S \cup \{v\}$ and each of the other agents receives a single leaf. No matter which subset $S$ is chosen, the only difference in the utility function is for the bundle $S \cup \{v\}$. Thus, the only way to find this allocation is to determine the set $S$ by querying the utility function with subsets of leaf vertices with size $n - 1$. The number of possible subsets of leaf vertices with size $n - 1$ is
    \[
    {2n - 2 \choose n - 1} \ge \left(\frac{2n - 2}{n - 1}\right)^{n - 1} = 2^{n - 1} \;.
    \]
    Assuming $S$ is chosen uniformly at random, on average it would take at least $2^{n - 2}$ queries to determine $S$.
\end{proof}

\begin{lemma}\label{lem:additive:bipartite:two:hard}
Even for a bipartite graph and two agents with identical additive utilities, computing a PMMS connected allocation is NP-hard.
\end{lemma}

\begin{proof}
    Follows directly from Theorem~5 in \citet{Greco:20}, which shows an equivalent result for MMS through a reduction from PARTITION. Since PMMS and MMS are equivalent when there are only two agents, the result also holds for PMMS.
\end{proof}

\section{Conclusion}
We introduced a novel local adaptation of PMMS into graph fair division. Our work made a major first step in addressing critical questions related to the existence and algorithmic aspects of PMMS within this context. We leave several important questions open. Most notably, the existence of a PMMS connected allocation is inevitably a challenging open problem due to the fact that this problem is at least as difficult as in the standard setting of fair division. Another direction is to consider a ``local" variant of other solution concepts, such as MMS. For instance, instead of pairwise comparison, agents may compare their bundles with all the bundles of their neighbors. 



\bibliographystyle{ACM-Reference-Format} 


\clearpage
\appendix

\section*{Appendix}

\section{Omitted material from Section~\ref{sec:twoagents}}\label{appendix:twoagents}

A utility function $u$ is \emph{submodular} if for any two subsets of items $X, Y \subseteq V$, it holds that $u(X) + u(Y) \ge u(X \cup Y) + u(X \cap Y)$. Utility function $u$ has \emph{binary marginal gains} if $u(X \cup \{v\}) - u(X) \in \{0, 1\}$ for every $X \subseteq V$ and $v \in V$. In other words, binary marginal gains is equivalent to that the utility of a bundle increases by $1$ or remains the same whenever a new item is introduced to the bundle.

\begin{proposition}\label{prop:PMMS:two:nonadditive}
Even for a cycle of length four and two agents with submodular utilities with binary marginal gains, a PMMS connected allocation may not exist.
\end{proposition}

\begin{proof}
    Number the items on the cycle in order from $1$ to $4$. Let $u_1$ be such that $u_1(\emptyset) = 0$, $u_1(B) = 1$ if $\{1, 2\} \not\subseteq B$ and $\{3, 4\} \not\subseteq B$ and $u_1(B) = 2$ otherwise. Similarly, let $u_2(\emptyset) = 0$, $u_2(B) = 1$ if $\{1, 4\} \not\subseteq B$ and $\{2, 3\} \not\subseteq B$ and $u_2(B) = 2$ otherwise. We can easily verify that both utility functions are submodular and have binary marginal gains. Moreover, the PMMS of both agents is 2, as $(\{1, 2\}, \{3, 4\})$ is the only PMMS partition of $u_1$ and $(\{1, 4\}, \{2, 3\})$ is the only PMMS partition of $u_2$. For both agents, $u_i(B) \ge 2$ holds only if $B$ contains one of the bundles in her PMMS partition. Since any bundle in the PMMS partition of $u_1$ overlaps with both bundles in the PMMS partition of $u_2$ and vice versa, there cannot exist any allocation in which both agents receive a bundle $B_i$ with $u_i(B_i) \ge 2$ and no PMMS allocation can exist.
\end{proof}

\section{Moving-knife style algorithm}\label{sec:moving-knife}

In this section, we present the moving-knife style algorithm for identical monotone utilities on a path, mentioned in \cref{sec:identical}. We will assume that the path consists of at least $n$ vertices. Otherwise, a PMMS connected allocation can trivially be found by providing as many agents as possible with one item.

The algorithm (\cref{alg:moving-knife}), works by in a discrete way moving $n - 1$ knives from one side of the path to the other. Whenever two adjacent bundles do not satisfy PMMS, the knife separating the two is moved.

\begin{algorithm}                      
\caption{Moving-knife style algorithm for constructing a PMMS connected allocation under identical utilities on a path}         
\label{alg:moving-knife}                          
\begin{algorithmic}[1]
\STATE Let $A = (\{1\}, \{2\}, \dots, \{n - 1\}, \{n, n + 1, \dots, k\})$ where the items on the path are numbered from $1$ to $k$ going from left to right.\label{line:moving-knife:initial}
\WHILE{$A$ is not PMMS}
    \STATE Let $i, j \in N$ be a pair of agents such that $i = j - 1$ and $u(A_i) < \PMMS(A_i \cup A_j)$.
    \STATE Transfer one and one outer item $g$ from $A_j$ to $A_i$ until $u(A_i) \ge \PMMS(A_i \cup A_j)$.\label{line:moving-knife:move}
\ENDWHILE
\end{algorithmic}
\end{algorithm}

Before we prove the correctness of the algorithm, we need the following result on PMMS partitions for monotone utilities on a path.

\begin{lemma}\label{lem:PMMS:path:extending-path}
    Let $u$ be a monotone utility function, $P$ a path and $P_1$ a subpath of $P$. Let $(A_1, A_2)$ be a PMMS partition of $P_1$ such that there is no PMMS partition $(A_3, A_4)$ of $P_1$ with $A_3 \subset A_1$ or $A_4 \subset A_1$. Let $A'_2$ be the bundle obtained by adding one or more items from $P \setminus P_1$ to $A_2$. Then, there is no PMMS partition $(B_1, B_2)$ of $A_1 \cup A'_2$ such that $B_1 \subset A_1$ or $B_2 \subset A_1$.
\end{lemma}

\begin{proof}
    Assume to the contrary that there exist such PMMS partition. Without loss of generality we assume that $B_1 \subset A_1$. Then it holds that
    \[\PMMS(A_1 \cup A_2) \le \PMMS(A_1 \cup A'_2) \le u(B_1) \]
    Since $u(A_2 \cup (A_1 \setminus B_1)) > u(A_2)$ it holds that $(B_1, A_2 \cup (A_1 \setminus B_1))$ is a PMMS partition of $P_1$, a contradiction.
\end{proof}

\Cref{lem:PMMS:path:extending-path} guarantees that if a knife is moved to the right by the minimal possible amount every time it needs to be moved, then no knife needs to be moved to the left at any point. We are now ready to show the correctness of the algorithm.

\begin{lemma}\label{lem:path-polynomial-p-mms}
    If the graph is a path and the agents have identical monotone utility functions, \cref{alg:moving-knife} finds a PMMS allocation in polynomial time.
\end{lemma}

\begin{proof}
    We wish to show that \cref{alg:moving-knife} finds a PMMS allocation in polynomial time. For this purpose, it suffice to show that for every pair of neighbouring agents $i, j \in N$ such that $i < j$ the following holds at any stage in the algorithm:
    \begin{enumerate}[label=(\roman*)]
        \item $i = j - 1$;\label{item:moving-knife:1}
        \item $u(A_j) \ge \PMMS(A_i \cup A_j)$;\label{item:moving-knife:2}
        \item either (1) $u(A_i) < \PMMS(A_i \cup A_j)$ or (2) there is no other PMMS partition of $A_i \cup A_j$ such that one of the bundles is contained in $A_i$, unless $|A_i| = 1$; and\label{item:moving-knife:3}
        \item $|A_i| > 0$ and $|A_j| > 0$.\label{item:moving-knife:4}
    \end{enumerate}

    In other words, we maintain the invariant that an agent always receives a bundle with utility at least PMMS when comparing to her neighbour on the left, while an agent comparing to her neighbour on the right either does not receive PMMS or has the minimal cardinality, non-empty bundle that guarantees the agent PMMS.
    
    We can easily verify that \ref{item:moving-knife:1}--\ref{item:moving-knife:4} hold in the initial allocation created on Line~\ref{line:moving-knife:initial}.

    We now show that transferring a good from $A_j$ to $A_i$ on Line~\ref{line:moving-knife:move} does not break any of \ref{item:moving-knife:1}--\ref{item:moving-knife:4}. Since a good is transferred from $A_j$ to $A_i$, \ref{item:moving-knife:1} holds as long as the last good in $A_j$ is not transferred, i.e., as long as \ref{item:moving-knife:4} holds. First, note that in any partition $(B_1, B_2)$ such that $|B_1| = 1$ or $|B_2| = 1$, the other bundle has utility at least $\PMMS(B_1 \cup B_2)$.
    Thus, $u(A_i) < \PMMS(A_i \cup A_j)$ cannot hold if $|A_j| = 1$ and \ref{item:moving-knife:4}, and consequently \ref{item:moving-knife:1}, hold after the transfer.
    
    For \ref{item:moving-knife:3}, notice that (1) always holds when a good is transferred. We claim that if (1) holds, then (2) must also hold. Indeed, utilities are monotone. Thus, for any subset $B \subseteq A_i$ it holds that $u(B) \le u(A_i) < \PMMS(A_i \cup A_j)$. Consequently, since the partition is not PMMS before the transfer, (2) must hold after the transfer and \ref{item:moving-knife:3} holds after the transfer. As a consequence, \ref{item:moving-knife:2} must also hold after the transfer. Otherwise, we claim that there cannot exist any PMMS partition of $A_i \cup A_j$, which must exist by \cref{lem:identical-utility-functions}. Indeed, since (2) held before the transfer, if $u(A_j) < \PMMS(A_i \cup A_j)$ after the transfer, there exists by monotonicity no PMMS partitions where one bundle is a subset of $A_i$ nor for $A_j$, a contradiction. Thus, \ref{item:moving-knife:2} holds after the transfer.
    
    Since \ref{item:moving-knife:4} guarantees that each agent has at least one item in their bundle and items are only transferred to agents with smaller indices, the algorithm must eventually either find a PMMS allocation or run out of items to transfer. However, by \ref{item:moving-knife:1}--\ref{item:moving-knife:4} if there is a pair of neighbouring agents for which PMMS is not satisfied, then a transfer can be made. Thus, the algorithm must eventually find a PMMS allocation, as there is a finite number of possible transfers of goods. In fact, we can only transfer each good $n - 1$ times, one for each neighbouring pair of agents.

    Since there is a polynomial number of iterations, the algorithm runs in polynomial time if each iteration can be performed in polynomial time. Since the PMMS of any neighbouring pair can be found by checking each of the $|A_i \cup A_j| \le k$ partitions of $A_i \cup A_j$ into connected bundles, we can easily verify that each iteration can be performed in polynomial time.
\end{proof}

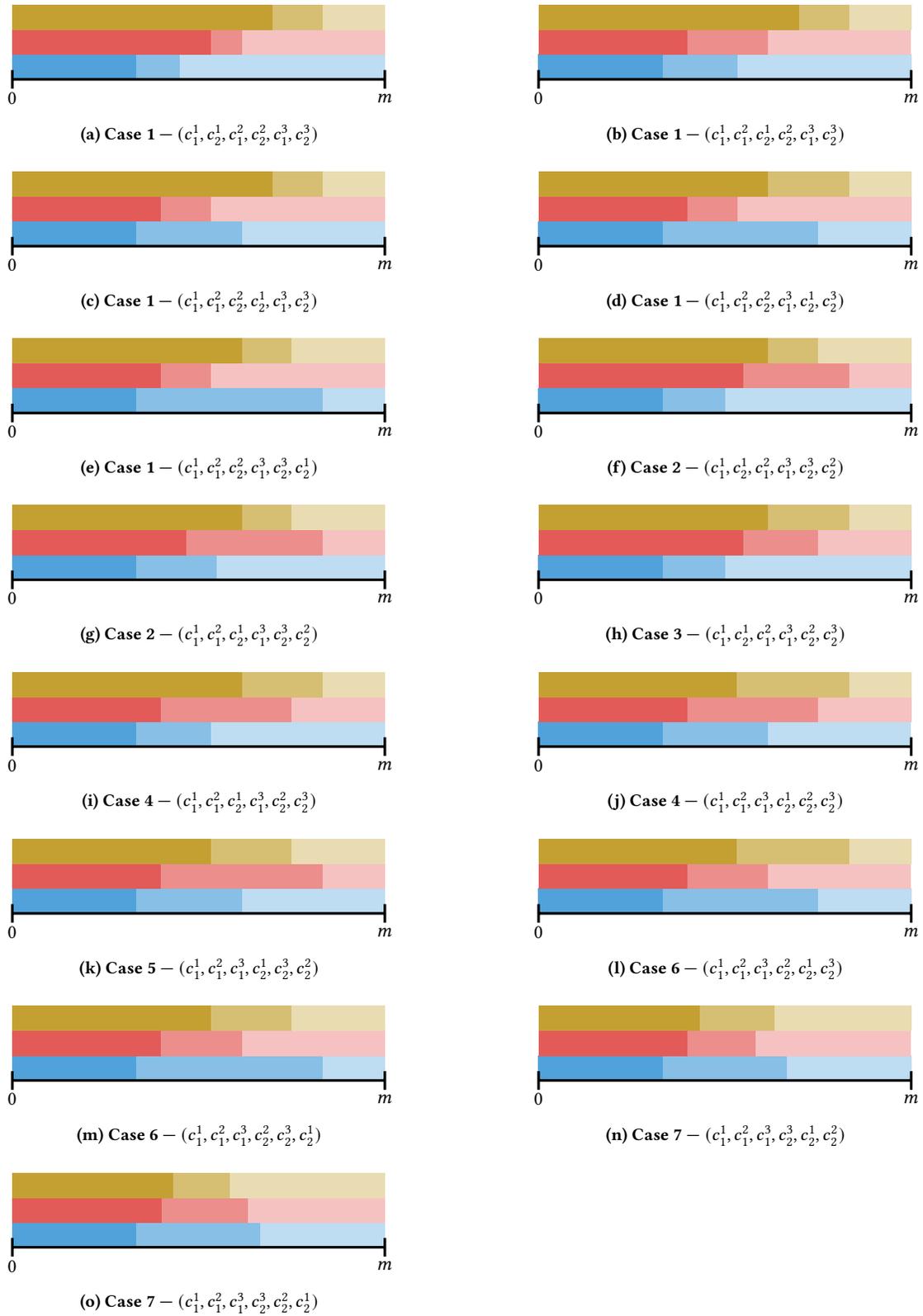
\begin{figure*}[t]
    \centering
    \begin{subfigure}{\columnwidth}
        \centering
        \begin{tikzpicture}
            \filldraw[fill=sblue!80, draw=none] (0, 0) rectangle (2, 0.4);
            \filldraw[fill=sblue!55, draw=none] (2, 0) rectangle (2.7, 0.4);
            \filldraw[fill=sblue!30, draw=none] (2.7, 0) rectangle (6, 0.4);
            
            \filldraw[fill=sred!80, draw=none] (0, 0.4) rectangle (3.2, 0.8);
            \filldraw[fill=sred!55, draw=none] (3.2, 0.4) rectangle (3.7, 0.8);
            \filldraw[fill=sred!30, draw=none] (3.7, 0.4) rectangle (6, 0.8);
        
            \filldraw[fill=syellow!80, draw=none] (0, 0.8) rectangle (4.2, 1.2);
            \filldraw[fill=syellow!55, draw=none] (4.2, 0.8) rectangle (5, 1.2);
            \filldraw[fill=syellow!30, draw=none] (5, 0.8) rectangle (6, 1.2);
            
            \draw[very thick] 
                (0, 0) -- (6, 0)
                (0, -0.15) -- (0, 0.15)
                (6, -0.15) -- (6, 0.15)
                ;
        
            \node at (0, -0.3) {$0$};
            \node at (6, -0.3) {$m$};
            
        \end{tikzpicture}
        \caption{Case 1 --- $(c_1^1, c_2^1, c_1^2, c_2^2, c_1^3, c_2^3)$}
    \end{subfigure}%
    \begin{subfigure}{\columnwidth}
        \centering
        \begin{tikzpicture}
            \filldraw[fill=sblue!80, draw=none] (0, 0) rectangle (2, 0.4);
            \filldraw[fill=sblue!55, draw=none] (2, 0) rectangle (3.2, 0.4);
            \filldraw[fill=sblue!30, draw=none] (3.2, 0) rectangle (6, 0.4);
            
            \filldraw[fill=sred!80, draw=none] (0, 0.4) rectangle (2.4, 0.8);
            \filldraw[fill=sred!55, draw=none] (2.4, 0.4) rectangle (3.7, 0.8);
            \filldraw[fill=sred!30, draw=none] (3.7, 0.4) rectangle (6, 0.8);
        
            \filldraw[fill=syellow!80, draw=none] (0, 0.8) rectangle (4.2, 1.2);
            \filldraw[fill=syellow!55, draw=none] (4.2, 0.8) rectangle (5, 1.2);
            \filldraw[fill=syellow!30, draw=none] (5, 0.8) rectangle (6, 1.2);
            
            \draw[very thick] 
                (0, 0) -- (6, 0)
                (0, -0.15) -- (0, 0.15)
                (6, -0.15) -- (6, 0.15)
                ;
        
            \node at (0, -0.3) {$0$};
            \node at (6, -0.3) {$m$};

        \end{tikzpicture}
        \caption{Case 1 --- $(c_1^1, c_1^2, c_2^1, c_2^2, c_1^3, c_2^3)$}
    \end{subfigure}\\%
    \vspace{\baselineskip}
    \begin{subfigure}{\columnwidth}
        \centering
        \begin{tikzpicture}
            \filldraw[fill=sblue!80, draw=none] (0, 0) rectangle (2, 0.4);
            \filldraw[fill=sblue!55, draw=none] (2, 0) rectangle (3.7, 0.4);
            \filldraw[fill=sblue!30, draw=none] (3.7, 0) rectangle (6, 0.4);
            
            \filldraw[fill=sred!80, draw=none] (0, 0.4) rectangle (2.4, 0.8);
            \filldraw[fill=sred!55, draw=none] (2.4, 0.4) rectangle (3.2, 0.8);
            \filldraw[fill=sred!30, draw=none] (3.2, 0.4) rectangle (6, 0.8);
        
            \filldraw[fill=syellow!80, draw=none] (0, 0.8) rectangle (4.2, 1.2);
            \filldraw[fill=syellow!55, draw=none] (4.2, 0.8) rectangle (5, 1.2);
            \filldraw[fill=syellow!30, draw=none] (5, 0.8) rectangle (6, 1.2);
            
            \draw[very thick] 
                (0, 0) -- (6, 0)
                (0, -0.15) -- (0, 0.15)
                (6, -0.15) -- (6, 0.15)
                ;
        
            \node at (0, -0.3) {$0$};
            \node at (6, -0.3) {$m$};

        \end{tikzpicture}
        \caption{Case 1 --- $(c_1^1, c_1^2, c_2^2, c_2^1, c_1^3, c_2^3)$}
    \end{subfigure}%
    \begin{subfigure}{\columnwidth}
        \centering
        \begin{tikzpicture}
            \filldraw[fill=sblue!80, draw=none] (0, 0) rectangle (2, 0.4);
            \filldraw[fill=sblue!55, draw=none] (2, 0) rectangle (4.5, 0.4);
            \filldraw[fill=sblue!30, draw=none] (4.5, 0) rectangle (6, 0.4);
            
            \filldraw[fill=sred!80, draw=none] (0, 0.4) rectangle (2.4, 0.8);
            \filldraw[fill=sred!55, draw=none] (2.4, 0.4) rectangle (3.2, 0.8);
            \filldraw[fill=sred!30, draw=none] (3.2, 0.4) rectangle (6, 0.8);
        
            \filldraw[fill=syellow!80, draw=none] (0, 0.8) rectangle (3.7, 1.2);
            \filldraw[fill=syellow!55, draw=none] (3.7, 0.8) rectangle (5, 1.2);
            \filldraw[fill=syellow!30, draw=none] (5, 0.8) rectangle (6, 1.2);
            
            \draw[very thick] 
                (0, 0) -- (6, 0)
                (0, -0.15) -- (0, 0.15)
                (6, -0.15) -- (6, 0.15)
                ;
        
            \node at (0, -0.3) {$0$};
            \node at (6, -0.3) {$m$};

        \end{tikzpicture}
        \caption{Case 1 --- $(c_1^1, c_1^2, c_2^2, c_1^3, c_2^1, c_2^3)$}
    \end{subfigure}\\%
    \vspace{\baselineskip}
    \begin{subfigure}{\columnwidth}
        \centering
        \begin{tikzpicture}
            \filldraw[fill=sblue!80, draw=none] (0, 0) rectangle (2, 0.4);
            \filldraw[fill=sblue!55, draw=none] (2, 0) rectangle (5, 0.4);
            \filldraw[fill=sblue!30, draw=none] (5, 0) rectangle (6, 0.4);
            
            \filldraw[fill=sred!80, draw=none] (0, 0.4) rectangle (2.4, 0.8);
            \filldraw[fill=sred!55, draw=none] (2.4, 0.4) rectangle (3.2, 0.8);
            \filldraw[fill=sred!30, draw=none] (3.2, 0.4) rectangle (6, 0.8);
        
            \filldraw[fill=syellow!80, draw=none] (0, 0.8) rectangle (3.7, 1.2);
            \filldraw[fill=syellow!55, draw=none] (3.7, 0.8) rectangle (4.5, 1.2);
            \filldraw[fill=syellow!30, draw=none] (4.5, 0.8) rectangle (6, 1.2);
            
            \draw[very thick] 
                (0, 0) -- (6, 0)
                (0, -0.15) -- (0, 0.15)
                (6, -0.15) -- (6, 0.15)
                ;
        
            \node at (0, -0.3) {$0$};
            \node at (6, -0.3) {$m$};
                                    
        \end{tikzpicture}
        \caption{Case 1 --- $(c_1^1, c_1^2, c_2^2, c_1^3, c_2^3, c_2^1)$}
    \end{subfigure}%
    \begin{subfigure}{\columnwidth}
        \centering
        \begin{tikzpicture}
            \filldraw[fill=sblue!80, draw=none] (0, 0) rectangle (2, 0.4);
            \filldraw[fill=sblue!55, draw=none] (2, 0) rectangle (3, 0.4);
            \filldraw[fill=sblue!30, draw=none] (3, 0) rectangle (6, 0.4);
            
            \filldraw[fill=sred!80, draw=none] (0, 0.4) rectangle (3.3, 0.8);
            \filldraw[fill=sred!55, draw=none] (3.3, 0.4) rectangle (5, 0.8);
            \filldraw[fill=sred!30, draw=none] (5, 0.4) rectangle (6, 0.8);
        
            \filldraw[fill=syellow!80, draw=none] (0, 0.8) rectangle (3.7, 1.2);
            \filldraw[fill=syellow!55, draw=none] (3.7, 0.8) rectangle (4.5, 1.2);
            \filldraw[fill=syellow!30, draw=none] (4.5, 0.8) rectangle (6, 1.2);
            
            \draw[very thick] 
                (0, 0) -- (6, 0)
                (0, -0.15) -- (0, 0.15)
                (6, -0.15) -- (6, 0.15)
                ;
        
            \node at (0, -0.3) {$0$};
            \node at (6, -0.3) {$m$};
            
        \end{tikzpicture}
        \caption{Case 2 --- $(c_1^1, c_2^1, c_1^2, c_1^3, c_2^3, c_2^2)$}
    \end{subfigure}\\%
    \vspace{\baselineskip}
    \begin{subfigure}{\columnwidth}
        \centering
        \begin{tikzpicture}
            \filldraw[fill=sblue!80, draw=none] (0, 0) rectangle (2, 0.4);
            \filldraw[fill=sblue!55, draw=none] (2, 0) rectangle (3.3, 0.4);
            \filldraw[fill=sblue!30, draw=none] (3.3, 0) rectangle (6, 0.4);
            
            \filldraw[fill=sred!80, draw=none] (0, 0.4) rectangle (2.8, 0.8);
            \filldraw[fill=sred!55, draw=none] (2.8, 0.4) rectangle (5, 0.8);
            \filldraw[fill=sred!30, draw=none] (5, 0.4) rectangle (6, 0.8);
        
            \filldraw[fill=syellow!80, draw=none] (0, 0.8) rectangle (3.7, 1.2);
            \filldraw[fill=syellow!55, draw=none] (3.7, 0.8) rectangle (4.5, 1.2);
            \filldraw[fill=syellow!30, draw=none] (4.5, 0.8) rectangle (6, 1.2);
            
            \draw[very thick] 
                (0, 0) -- (6, 0)
                (0, -0.15) -- (0, 0.15)
                (6, -0.15) -- (6, 0.15)
                ;
        
            \node at (0, -0.3) {$0$};
            \node at (6, -0.3) {$m$};
            
        \end{tikzpicture}
        \caption{Case 2 --- $(c_1^1, c_1^2, c_2^1, c_1^3, c_2^3, c_2^2)$}
    \end{subfigure}%
    \begin{subfigure}{\columnwidth}
        \centering
        \begin{tikzpicture}
            \filldraw[fill=sblue!80, draw=none] (0, 0) rectangle (2, 0.4);
            \filldraw[fill=sblue!55, draw=none] (2, 0) rectangle (3.0, 0.4);
            \filldraw[fill=sblue!30, draw=none] (3.0, 0) rectangle (6, 0.4);
            
            \filldraw[fill=sred!80, draw=none] (0, 0.4) rectangle (3.3, 0.8);
            \filldraw[fill=sred!55, draw=none] (3.3, 0.4) rectangle (4.5, 0.8);
            \filldraw[fill=sred!30, draw=none] (4.5, 0.4) rectangle (6, 0.8);
        
            \filldraw[fill=syellow!80, draw=none] (0, 0.8) rectangle (3.7, 1.2);
            \filldraw[fill=syellow!55, draw=none] (3.7, 0.8) rectangle (5, 1.2);
            \filldraw[fill=syellow!30, draw=none] (5, 0.8) rectangle (6, 1.2);
            
            \draw[very thick] 
                (0, 0) -- (6, 0)
                (0, -0.15) -- (0, 0.15)
                (6, -0.15) -- (6, 0.15)
                ;
        
            \node at (0, -0.3) {$0$};
            \node at (6, -0.3) {$m$};
            
        \end{tikzpicture}
        \caption{Case 3 --- $(c_1^1, c_2^1, c_1^2, c_1^3, c_2^2, c_2^3)$}
    \end{subfigure}\\%
    \vspace{\baselineskip}%
    \begin{subfigure}{\columnwidth}
        \centering
        \begin{tikzpicture}
            \filldraw[fill=sblue!80, draw=none] (0, 0) rectangle (2, 0.4);
            \filldraw[fill=sblue!55, draw=none] (2, 0) rectangle (3.2, 0.4);
            \filldraw[fill=sblue!30, draw=none] (3.2, 0) rectangle (6, 0.4);
            
            \filldraw[fill=sred!80, draw=none] (0, 0.4) rectangle (2.4, 0.8);
            \filldraw[fill=sred!55, draw=none] (2.4, 0.4) rectangle (4.5, 0.8);
            \filldraw[fill=sred!30, draw=none] (4.5, 0.4) rectangle (6, 0.8);
        
            \filldraw[fill=syellow!80, draw=none] (0, 0.8) rectangle (3.7, 1.2);
            \filldraw[fill=syellow!55, draw=none] (3.7, 0.8) rectangle (5, 1.2);
            \filldraw[fill=syellow!30, draw=none] (5, 0.8) rectangle (6, 1.2);
            
            \draw[very thick] 
                (0, 0) -- (6, 0)
                (0, -0.15) -- (0, 0.15)
                (6, -0.15) -- (6, 0.15)
                ;
        
            \node at (0, -0.3) {$0$};
            \node at (6, -0.3) {$m$};
            
        \end{tikzpicture}
        \caption{Case 4 --- $(c_1^1, c_1^2, c_2^1, c_1^3, c_2^2, c_2^3)$}
    \end{subfigure}%
    \begin{subfigure}{\columnwidth}
        \centering
        \begin{tikzpicture}
            \filldraw[fill=sblue!80, draw=none] (0, 0) rectangle (2, 0.4);
            \filldraw[fill=sblue!55, draw=none] (2, 0) rectangle (3.7, 0.4);
            \filldraw[fill=sblue!30, draw=none] (3.7, 0) rectangle (6, 0.4);
            
            \filldraw[fill=sred!80, draw=none] (0, 0.4) rectangle (2.4, 0.8);
            \filldraw[fill=sred!55, draw=none] (2.4, 0.4) rectangle (4.5, 0.8);
            \filldraw[fill=sred!30, draw=none] (4.5, 0.4) rectangle (6, 0.8);
        
            \filldraw[fill=syellow!80, draw=none] (0, 0.8) rectangle (3.2, 1.2);
            \filldraw[fill=syellow!55, draw=none] (3.2, 0.8) rectangle (5, 1.2);
            \filldraw[fill=syellow!30, draw=none] (5, 0.8) rectangle (6, 1.2);
            
            \draw[very thick] 
                (0, 0) -- (6, 0)
                (0, -0.15) -- (0, 0.15)
                (6, -0.15) -- (6, 0.15)
                ;
        
            \node at (0, -0.3) {$0$};
            \node at (6, -0.3) {$m$};
            
        \end{tikzpicture}
        \caption{Case 4 --- $(c_1^1, c_1^2, c_1^3, c_2^1, c_2^2, c_2^3)$}
    \end{subfigure}\\%
    \vspace{\baselineskip}%
    \begin{subfigure}{\columnwidth}
        \centering
        \begin{tikzpicture}
            \filldraw[fill=sblue!80, draw=none] (0, 0) rectangle (2, 0.4);
            \filldraw[fill=sblue!55, draw=none] (2, 0) rectangle (3.7, 0.4);
            \filldraw[fill=sblue!30, draw=none] (3.7, 0) rectangle (6, 0.4);
            
            \filldraw[fill=sred!80, draw=none] (0, 0.4) rectangle (2.4, 0.8);
            \filldraw[fill=sred!55, draw=none] (2.4, 0.4) rectangle (5, 0.8);
            \filldraw[fill=sred!30, draw=none] (5, 0.4) rectangle (6, 0.8);
        
            \filldraw[fill=syellow!80, draw=none] (0, 0.8) rectangle (3.2, 1.2);
            \filldraw[fill=syellow!55, draw=none] (3.2, 0.8) rectangle (4.5, 1.2);
            \filldraw[fill=syellow!30, draw=none] (4.5, 0.8) rectangle (6, 1.2);
            
            \draw[very thick] 
                (0, 0) -- (6, 0)
                (0, -0.15) -- (0, 0.15)
                (6, -0.15) -- (6, 0.15)
                ;
        
            \node at (0, -0.3) {$0$};
            \node at (6, -0.3) {$m$};
            
        \end{tikzpicture}
        \caption{Case 5 --- $(c_1^1, c_1^2, c_1^3, c_2^1, c_2^3, c_2^2)$}
    \end{subfigure}%
    \begin{subfigure}{\columnwidth}
        \centering
        \begin{tikzpicture}
            \filldraw[fill=sblue!80, draw=none] (0, 0) rectangle (2, 0.4);
            \filldraw[fill=sblue!55, draw=none] (2, 0) rectangle (4.5, 0.4);
            \filldraw[fill=sblue!30, draw=none] (4.5, 0) rectangle (6, 0.4);
            
            \filldraw[fill=sred!80, draw=none] (0, 0.4) rectangle (2.4, 0.8);
            \filldraw[fill=sred!55, draw=none] (2.4, 0.4) rectangle (3.7, 0.8);
            \filldraw[fill=sred!30, draw=none] (3.7, 0.4) rectangle (6, 0.8);
        
            \filldraw[fill=syellow!80, draw=none] (0, 0.8) rectangle (3.2, 1.2);
            \filldraw[fill=syellow!55, draw=none] (3.2, 0.8) rectangle (5, 1.2);
            \filldraw[fill=syellow!30, draw=none] (5, 0.8) rectangle (6, 1.2);
            
            \draw[very thick] 
                (0, 0) -- (6, 0)
                (0, -0.15) -- (0, 0.15)
                (6, -0.15) -- (6, 0.15)
                ;
        
            \node at (0, -0.3) {$0$};
            \node at (6, -0.3) {$m$};
            
        \end{tikzpicture}
        \caption{Case 6 --- $(c_1^1, c_1^2, c_1^3, c_2^2, c_2^1, c_2^3)$}
    \end{subfigure}\\%
    \vspace{\baselineskip}%
    \begin{subfigure}{\columnwidth}
        \centering
        \begin{tikzpicture}
            \filldraw[fill=sblue!80, draw=none] (0, 0) rectangle (2, 0.4);
            \filldraw[fill=sblue!55, draw=none] (2, 0) rectangle (5, 0.4);
            \filldraw[fill=sblue!30, draw=none] (5, 0) rectangle (6, 0.4);
            
            \filldraw[fill=sred!80, draw=none] (0, 0.4) rectangle (2.4, 0.8);
            \filldraw[fill=sred!55, draw=none] (2.4, 0.4) rectangle (3.7, 0.8);
            \filldraw[fill=sred!30, draw=none] (3.7, 0.4) rectangle (6, 0.8);
        
            \filldraw[fill=syellow!80, draw=none] (0, 0.8) rectangle (3.2, 1.2);
            \filldraw[fill=syellow!55, draw=none] (3.2, 0.8) rectangle (4.5, 1.2);
            \filldraw[fill=syellow!30, draw=none] (4.5, 0.8) rectangle (6, 1.2);
            
            \draw[very thick] 
                (0, 0) -- (6, 0)
                (0, -0.15) -- (0, 0.15)
                (6, -0.15) -- (6, 0.15)
                ;
        
            \node at (0, -0.3) {$0$};
            \node at (6, -0.3) {$m$};
            
        \end{tikzpicture}
        \caption{Case 6 --- $(c_1^1, c_1^2, c_1^3, c_2^2, c_2^3, c_2^1)$}
    \end{subfigure}%
    \begin{subfigure}{\columnwidth}
        \centering
        \begin{tikzpicture}
            \filldraw[fill=sblue!80, draw=none] (0, 0) rectangle (2, 0.4);
            \filldraw[fill=sblue!55, draw=none] (2, 0) rectangle (4, 0.4);
            \filldraw[fill=sblue!30, draw=none] (4, 0) rectangle (6, 0.4);
            
            \filldraw[fill=sred!80, draw=none] (0, 0.4) rectangle (2.4, 0.8);
            \filldraw[fill=sred!55, draw=none] (2.4, 0.4) rectangle (3.5, 0.8);
            \filldraw[fill=sred!30, draw=none] (3.5, 0.4) rectangle (6, 0.8);
    
            \filldraw[fill=syellow!80, draw=none] (0, 0.8) rectangle (2.6, 1.2);
            \filldraw[fill=syellow!55, draw=none] (2.6, 0.8) rectangle (3.8, 1.2);
            \filldraw[fill=syellow!30, draw=none] (3.8, 0.8) rectangle (6, 1.2);
            
            \draw[very thick] 
                (0, 0) -- (6, 0)
                (0, -0.15) -- (0, 0.15)
                (6, -0.15) -- (6, 0.15)
                ;
    
            \node at (0, -0.3) {$0$};
            \node at (6, -0.3) {$m$};
        \end{tikzpicture}
        \caption{Case 7 --- $(c_1^1, c_1^2, c_1^3, c_2^3, c_2^1, c_2^2)$}
    \end{subfigure}\\%
    \vspace{\baselineskip}%
    \begin{subfigure}{\columnwidth}
        \centering
        \begin{tikzpicture}
            \filldraw[fill=sblue!80, draw=none] (0, 0) rectangle (2, 0.4);
            \filldraw[fill=sblue!55, draw=none] (2, 0) rectangle (4, 0.4);
            \filldraw[fill=sblue!30, draw=none] (4, 0) rectangle (6, 0.4);
            
            \filldraw[fill=sred!80, draw=none] (0, 0.4) rectangle (2.4, 0.8);
            \filldraw[fill=sred!55, draw=none] (2.4, 0.4) rectangle (3.8, 0.8);
            \filldraw[fill=sred!30, draw=none] (3.8, 0.4) rectangle (6, 0.8);
    
            \filldraw[fill=syellow!80, draw=none] (0, 0.8) rectangle (2.6, 1.2);
            \filldraw[fill=syellow!55, draw=none] (2.6, 0.8) rectangle (3.5, 1.2);
            \filldraw[fill=syellow!30, draw=none] (3.5, 0.8) rectangle (6, 1.2);
            
            \draw[very thick] 
                (0, 0) -- (6, 0)
                (0, -0.15) -- (0, 0.15)
                (6, -0.15) -- (6, 0.15)
                ;
    
            \node at (0, -0.3) {$0$};
            \node at (6, -0.3) {$m$};
        \end{tikzpicture}
        \caption{Case 7 --- $(c_1^1, c_1^2, c_1^3, c_2^3, c_2^2, c_2^1)$}
    \end{subfigure}%
    \begin{subfigure}{\columnwidth}
        \phantom{0}
    \end{subfigure}

    \caption{Visualisation of all 15 permutations of cut orders in \cref{thr:three-agents-path} and the 7 cases that handle them.}
    \label{fig:cases-path}
\end{figure*}
\end{document}